\documentclass[letterpaper]{article}
\usepackage[a4paper, margin=1.5in]{geometry}

\newif\ifarxiv
\arxivfalse

\usepackage{microtype}
\usepackage{graphicx}
\usepackage{subcaption}
\usepackage{booktabs} 
\usepackage{yhmath}
\usepackage{tabularx}
\usepackage{comment}
\usepackage{enumitem}
\usepackage{hyperref}
\hypersetup{colorlinks = true, citecolor=blue}
\usepackage[table]{xcolor}

\usepackage[algo2e,noend,ruled,linesnumbered]{algorithm2e} 

\usepackage{amsmath}
\usepackage{amssymb}
\usepackage{mathtools}
\usepackage{amsthm}
\usepackage{dsfont}
\usepackage{enumitem}

\usepackage[capitalize,noabbrev]{cleveref}
\theoremstyle{plain}
{\itshape}{\rmfamily}
{\bfseries}{\itshape}

\theoremstyle{plain}
\newtheorem{lemma}{Lemma}
\theoremstyle{plain}

\theoremstyle{plain}
\newtheorem{definition}{Definition}
\theoremstyle{plain}
\newtheorem{theorem}{Theorem}

\usepackage[disable]{todonotes}

\newcommand{\Oh}{\mathcal{O}}

\newcommand{\vi}{\ensuremath{\mathsf{vi}}\xspace}
\newcommand{\tdp}{\ensuremath{\mathsf{tdp}}\xspace}
\newcommand{\nd}{\ensuremath{\mathsf{nd}}\xspace}
\newcommand{\cc}{\ensuremath{\mathsf{cc}}\xspace}

\newcommand{\vc}{\ensuremath{\mathsf{vc}}\xspace}
\newcommand{\bdd}[1][1]{\ensuremath{{#1}\textnormal{-}\mathsf{bdd}}\xspace}
\newcommand{\pvc}[1][4]{\ensuremath{{#1}\textnormal{-}\mathsf{pvc}}\xspace}
\newcommand{\core}[1][2]{\ensuremath{{#1}\textnormal{-}\mathsf{core}}\xspace}
\newcommand{\fvs}{\ensuremath{\mathsf{fvs}}\xspace}
\newcommand{\cvd}{\ensuremath{\mathsf{cvd}}\xspace}
\newcommand{\dco}{\ensuremath{\mathsf{dco}}\xspace}
\newcommand{\coc}[1][\ell]{\ensuremath{{#1}\textnormal{-}\mathsf{coc}}\xspace}
\newcommand{\tw}{\ensuremath{\mathsf{tw}}\xspace}
\newcommand{\sw}{\ensuremath{\mathsf{sw}}\xspace}
\newcommand{\mw}{\ensuremath{\mathsf{mw}}\xspace}
\newcommand{\md}{\ensuremath{\mathsf{md}}\xspace}
\newcommand{\conn}{\ensuremath{\mathsf{conn}}\xspace}
\newcommand{\cut}{\ensuremath{\mathsf{cut}}\xspace}

\newcommand{\mT}{\mathcal{T}}
\newcommand{\mv}{\mathcal{V}}

\newcommand{\poly}{\ensuremath{\operatorname{poly}}}

\usepackage{scrextend}          

\usetikzlibrary{patterns}

\newcommand{\mc}{\mathcal{C}}

\usepackage{etoolbox}
\newcommand{\appsymb}{$\bigstar$}
\newcommand{\appref}[1]{\hyperref[proof:#1]{\appsymb}}
\DeclareMathOperator{\cl}{cl}

\usepackage{authblk}
\usepackage{orcidlink}
\begin{document}

\title{The Parameter Report:\\ An Orientation Guide for Data-Driven Parameterization}
\pagestyle{plain}

\author[1]{Christian Komusiewicz\orcidlink{0000-0003-0829-7032}}
\author[1]{Nils Morawietz\orcidlink{0000-0002-7283-4982}}
\author[1]{Frank~Sommer\thanks{Supported by the Alexander von Humboldt Foundation.}\orcidlink{0000-0003-4034-525X}}
\author[1]{Luca Pascal Staus\thanks{Supported by the Carl Zeiss Foundation, Germany, 
within the project ‘‘Interactive Inference’’.}\orcidlink{0009-0004-3020-1011}}

\affil[1]{\small Institute of Computer Science, Friedrich Schiller University Jena, Germany}

\date{}

\maketitle

\begin{abstract}
A strength of parameterized algorithmics is that each problem can be parameterized by an essentially inexhaustible set of parameters. Usually, the choice of the considered parameter is informed by the theoretical relations between parameters with the general goal of achieving FPT-algorithms for smaller and smaller parameters. However, the FPT-algorithms for smaller parameters usually have higher running times and it is unclear whether the decrease in the parameter value or the increase in the running time bound dominates in real-world data. This question cannot be answered from purely theoretical considerations and any answer requires knowledge on typical parameter values.
  
To provide a data-driven guideline for parameterized complexity studies of graph problems, we present the first comprehensive comparison of parameter values for a set of benchmark graphs originating from real-world applications. Our study covers degree-related parameters, such as maximum degree or degeneracy, neighborhood-based parameters such as neighborhood diversity and modular-width, modulator-based parameters such as vertex cover number and feedback vertex set number, and the treewidth of the graphs. 

Our results may help assess the significance of FPT-running time bounds on the solvability of real-world instances. For example, the vertex cover number~$\vc$ of~$n$-vertex graphs is often only slightly below~$n/2$. Thus, a running time bound of~$\Oh(2^{\vc})$ is only slightly better than a running time bound of~$\Oh(1.4^{n})$. In contrast, the treewidth~$\tw$ is almost always below~$n/3$ and often close to~$n/10$, making a running time of~$\Oh(2^{\tw})$ much more practical on real-world instances.

We make our implementation and full experimental data openly available\footnote{The source code of all solvers, all graph data and our experimental results are all publicly available at \url{https://www.fmi.uni-jena.de/en/19723/parameter-report}.}. In particular, this provides the first implementations for several graph parameters such as 4-path vertex cover number and vertex integrity. 
\end{abstract}
\section{Introduction}
\label{sec:intro}

In the quest to cope with computational intractability, parameterized algorithmics offers an approach for designing exact algorithms with practically useful running time guarantees. The aim is to design FPT-algorithms, that is, algorithms with running time~$f(k)\cdot \poly(n)$ where~$k$ is a problem-specific parameter and~$n$ denotes the overall input size. 
Now, if~$k$ is small on real-world instances and~$f$ does not grow too quickly, then these algorithms have a practically feasible running time. The classic example for this approach is the NP-hard \textsc{Vertex Cover} problem where the input is a graph~$G$ and a number~$k$ and the task is to decide whether~$G$ contains a set of at most~$k$ vertices that cover all edges. \textsc{Vertex Cover} can be easily seen to be solvable in~$2^k\cdot \poly(n)$~time by a simple branching algorithm. After achieving such a first FPT-algorithm for a problem at hand, the daily trade of researchers in parameterized algorithms is now to chip away from the rock of intractability by providing better and better FPT-algorithms~\cite{DF99}.

There are two pathways to achieve this: The first is to improve FPT-algorithms for the parameter~$k$ at hand which means first and foremost to improve the superpolynomial running time part~$f(k)$. For \textsc{Vertex Cover} this pathway has been thoroughly explored, with the current best FPT-algorithm having a running time of~$\Oh(1.26^k + n)$~\cite{DN24}.  The second pathway is to consider smaller and smaller parameters: After developing an FPT-algorithm for some parameter~$k$, one may try to find an FPT-algorithm for some parameter~$k'$ which is never larger than~$k$. For \textsc{Vertex Cover}, such a parameter is for example the treewidth~$\tw$ of~$G$ for which an algorithm with running time~$\Oh(2^{\tw}\cdot n)$ is known.\footnote{Here, one assumes that one is given a tree decomposition of width~$\tw$.} From a purely theoretical standpoint, the latter algorithm is usually deemed preferable: For \textsc{Vertex Cover}, the parameterization by the solution size bound~$k$ is essentially a parameterization by the vertex cover number of~$G$, denoted~$\vc$. This number is the size of a smallest vertex cover of~$G$. Due to the known parameter relation~$\vc\ge \tw$, any FPT-algorithm for the treewidth trivially implies an FPT-algorithm for the vertex cover number~\vc. 
This line of research in parameterized algorithms is commonly referred to as multivariate algorithmics~\cite{BHKN14,FJR13,FLM+09,Kom16,KN12} or structural parameterization~\cite{BCJS16,GKO21,HKNS15,JK11}. There are at least two potential pitfalls with this approach: First, at least for graph problems, the space of known parameterizations is huge and ever-growing. Hence, it is infeasible to thoroughly study all parameterizations and essentially every choice of a collection of parameters in an algorithmic study is in some sense arbitrary. Second, the improvement in the parameter dimension often comes at the cost of an increased running time. In the \textsc{Vertex Cover} example, the running time bound for treewidth is larger than for the vertex cover number. Hence, it is not clear whether the theoretical improvement in the parameter dimension carries over to better algorithms in practice or whether the running times actually get worse. 

To avoid both pitfalls, we propose a data-driven approach, where the actual values of parameters in real-world instances guide the navigation through parameter space: When going from a parameterization~$k$ to a parameterization~$k'$, we may then choose the parameterization which has smaller values in practice so any obtained FPT-algorithms will be more relevant in practice. Moreover, when knowing the typical relation between the values of two parameters in real-world instances, we may better compare known running time bounds.

\paragraph{Our Contribution.} To make a first step towards this data-driven approach, we present a study on the values of 21 popular graph parameters on a benchmark set of 144 graphs from real-world applications. The~$21$ parameters are as follows:
\begin{itemize}
	\item Degree-related: maximum degree ($\Delta$), $h$-index, degeneracy ($d$), $\core[2]$, and $\core[3]$, \hbox{$c$-closure} ($c$), and weak-$c$-closure ($\gamma$).
	\item Neighborhood-based: neighborhood diversity ($\nd$), Dilworth number ($\nabla$), modular-width ($\mw$), and splitwidth ($\sw$).
	\item Modulator-based: vertex cover number  ($\vc$), $1$-bounded-degree deletion number  ($\bdd[1]$), $2$-bounded-degree deletion number  ($\bdd[2]$), $4$-path vertex cover number  ($\pvc[4]$),  cluster vertex deletion number ($\cvd$), distance to cograph ($\dco$), vertex-integrity~$(\vi)$, and feedback vertex set number ($\fvs$).
	\item Treewidth ($\tw$) and treedepth~($\tdp$).
\end{itemize}

We present various ways to compare the parameter values in an algorithmic setting and describe some scenarios how the parameter values may be used to guide the exploration of parameter space. Our results indicate, for example, that the running time bound of~$\Oh(2^\tw\cdot n)$ for the \textsc{Vertex Cover} problem is better than the~$\Oh(1.26^\vc+n)$-time bound.  We make the source code of the implementation of the algorithms for computing the considered graph parameters and all the experimental data fully available\footnote{For several parameters, for example for the treewidth, we use the currently best openly available solvers.} to facilitate the computation of parameter values by the parameterized algorithmics community. Since the exact computation of some parameters is hard, we split our evaluation for these parameters in two parts: First, we consider those 59 instances where all parameters can be computed exactly. For the remaining 85 instances, we resort to heuristically computed parameter values whenever the exact algorithms failed.  

Moreover, we analyze the parameter values that can be obtained via two approaches for fusing neighborhood diversity~$\nd$~\cite{Lampis20} and modular-width~$\mw$~\cite{BJ00,HK23,KM22} with essentially any other parameter~$k$ which gives a parameter that is simultaneously upper-bounded by~$\nd$ and~$k$ and~$\mw$ and~$k$, respectively. 


This work is structured as follows. In \Cref{sec:prelim} we introduce some graph notation and define all graph parameters under consideration. In \Cref{sec:eval} we present the evaluation for the exact parameter values. In \Cref{sec:eval-heuristic} we present the evaluation for the heuristically computed parameter values.  
Some details concerning the algorithms for computing the parameters are given in~\Cref{sec:algos}.

\section{Preliminaries}
\label{sec:prelim}

\subsection{Graph Notation}

For~$x\in\mathds{N}$ by~$[x]$ we denote the set~$\{1,2,\ldots, x\}$.
We use standard graph notation according to Diestel~\cite{Diestel16}.
An \emph{undirected graph} is a tuple~$G \coloneqq  (V, E)$ where~$V$ is the set of \emph{vertices} and~$E\subseteq \{\{u,v\}\subseteq V: u\ne v\}$
is the set of \emph{edges}.
We also use~$V(G)$ and~$E(G)$ to denote the vertices and edges of~$G$, respectively.
We let~$n\coloneqq |V(G)|$ and~$m\coloneqq |E(G)|$.
For a vertex~$v$ of~$G$, we denote by~$N_G(v) \coloneqq \{w\in V(G): \{v,w\}\in E(G)\}$ the~\emph{open neighborhood} of~$v$.
Let~$X,Z\subseteq V(G)$.
By~$N_G(X)\coloneqq (\cup_{v\in X} N_G(v)) \setminus X$ we denote the \emph{open neighborhood of~$X$} and by~$N_G[X]\coloneqq N_G(X)\cup X$ we denote the \emph{closed neighborhood of~$X$}.
By~$E_G(X,Z)\coloneqq \{\{x,z\}\in E(G) : x\in X \text{ and } z\in Z\}$ we denote the set of edges \emph{between}~$X$ and~$Z$. 
Moreover, we use~$E_G(X)$ as a shorthand for~$E_G(X,X)$.
We omit the subscript if~$G$ is clear from context.
By~$G[X]\coloneqq (X,E_G(X))$ we denote the \emph{subgraph of~$G$ induced by}~$X$.
A graph~$G'\coloneqq (V',E')$ is a \emph{subgraph} of~$G$ if~$V'\subseteq V(G)$ and~$E'\subseteq E(G[V'])$.
We let~$G - X\coloneqq G[V\setminus X]$ denote the subgraph obtained by \emph{removing} the vertices of~$X$ from~$G$.

A sequence~$P\coloneqq (v_1,\dots, v_p)$ of vertices is a \emph{path} in~$G$ if~$\{v_i,v_{i+1}\}\in E(G)$ for each~$i\in[p-1]$ and no vertex appears more than once.
Its \emph{length} is~$p-1$.
For a path~$P\coloneqq (v_1,\dots, v_p)$ in~$G$, we denote by~$V(P) \coloneqq \{v_i: i\in [p]\}$ the vertices of~$P$ and by~$E(P) \coloneqq \{\{v_i,v_{i+1}\} : i\in [p-1]\}$.
Furthermore, a path~$P$ in~$G$ is an~\emph{induced path} if~$E(G[V(P)])=E(P)$.
Two vertices~$u$ and~$v$ are \emph{connected} if there exists a path~$P$ such that~$u,v\in V(P)$.
A \emph{connected component} of~$G$ is an inclusion maximal induced subgraph of~$G$ where any two vertices are connected to each other.
Let~$G$ be a graph and let~$X\subseteq V(G)$. 
By~$\conn(X)$ we denote the \emph{vertex connectivity} of the graph~$G[X]$, that is, $\conn(X)$ is the smallest number~$r$, such that for some~$Z\subseteq X$ of size~$r$, the graph~$G[X\setminus Z]$ has more than one connected component.
\todo{Move defs that are used only in appendix there? E.g. $\conn$,simplicial} 
A vertex set~$X$ is a \emph{clique} if each two distinct vertices in~$X$ are adjacent. A vertex~$v$ is \emph{simplicial} if~$N(v)$ is a clique.

\subsection{Graph Parameters}

\subsubsection{Degree-Related}

Let~$v \in V(G)$.
We denote the \emph{degree} of~$v$ by~$\deg_G(v)\coloneqq |N(v)|$.
The \emph{maximum degree} and \emph{minimum degree} of $G$ are~$\Delta(G) \coloneqq  \max_{v \in V(G)} \deg_G(v)$ and~$\delta(G) \coloneqq  \min_{v \in V(G)} \deg_G(v)$, respectively.
The \emph{degeneracy} of~$G$ is $d(G) \coloneqq  \max_{S \subseteq V(G)} \delta_{G[S]}$~\cite{LW70}.
The~\emph{${h}$-index} of a graph~$G$, denoted~$h(G)$, is the largest integer~$h$ such that~$G$ has at least~$h$ vertices of degree at least~$h$~\cite{ES12}.
The \emph{closure number}~$\cl_G(v)$ of a vertex~$v$ is~$\max_{u\in V(G)\setminus N[v]}|N(v)\cap N(u)|$.
We say that~$G$ is \emph{$c$-closed} if~$\cl_G(v)< c$ for each vertex~$v\in V(G)$.
Furthermore, the \emph{closure number} of a graph $G$, denoted~$\cl(G)$, is the smallest integer~$c$ such that~$G$
is $c$-closed~\cite{FRSWW20}.
We say that~$G$ is weakly~$\gamma$-closed if every induced subgraph~$G'$ of~$G$ has a vertex~$v\in V(G')$ such that~$\cl_{G'}(v)<\gamma$.
The \emph{weak closure number} of a graph $G$, denoted~$\gamma(G)$ is the smallest integer~$\gamma$ such that~$G$ is weakly~$\gamma$-closed~\cite{FRSWW20}.

The \emph{$\core[k]$} of~$G$ is the maximal induced subgraph~$H$ of~$G$ such that~$\delta(H) \ge k$~\cite{Seidman83}.
For~$k\le d$, the $\core[k]$ is nonempty and for~$k>d$, the $\core[k]$ has size~0. We consider the cases~$k=2$ and~$k=3$ and let~$\core[2]$ and~$\core[3]$ denote the number of vertices of the corresponding cores.

\subsubsection{Neighborhood-Based Parameters}
Here, we define the considered parameters which aim at quantifying the complexity of the family of different neighborhoods in a graph.

Two vertices~$u,w\in V(G)$ have the \emph{same type} if~$N(u)\setminus\{w\}=N(w)\setminus\{u\}$. Vertices of the same type are also called \emph{twins}.
The \emph{neighborhood diversity}~$\nd(G)$ of~$G$ is the smallest number~$\ell$ such that there is a partition~$(V_1, \ldots, V_\ell)$ of~$V(G)$, where in each subset~$V_i$, all vertices have the same type~\cite{Lampis12}.

The \emph{Dilworth number}~$\nabla(G)$ of~$G$ is the size of a largest set of vertices~$D$ such that for each two vertices~$u$ and~$w$ in $D$ we have~$N(u)\not\subseteq N(w)$ and~$N(w)\not\subseteq N(u)$~\cite{Dilworth50,BKS+24}.

\paragraph{Modular-width.}
A \emph{modular decomposition} of a graph~$G=(V,E)$ is a pair~$(\mT,\beta)$ consisting of a rooted tree~$\mT=(\mv, \mathcal{A})$ and a function~$\beta$ that maps each node~$x\in \mv$ to a graph~$\beta(x)$.
If~$x$ is a leaf of~$\mT$, then~$\beta(x)$ contains a single vertex of~$V$ and for each vertex~$v\in V$, there is exactly one leaf~$\ell$ of~$\mathcal{T}$ such that the graph~$\beta(\ell)$ consists only of~$v$.
If~$x$ is not a leaf node, then the vertex set of~$\beta(x)$ is exactly the set of child nodes of~$x$ in~$\mathcal{T}$.
Moreover, let~$V_x$ denote the set of vertices of~$V$ contained in leaf nodes of the subtree rooted in~$x$.
Formally, $V_x$ is defined as~$V(\beta(\ell))$ for leaf nodes~$\ell$ and recursively defined as~$\bigcup_{y\in V(\beta(x))} V_y$ for each nonleaf node~$x$.
Moreover, we define~$G_x = (V_x,E_x) \coloneqq   G[V_x]$.
A modular decomposition has the property that for each nonleaf node~$x$ and any pair of distinct nodes~$y\in V(\beta(x))$ and~$z\in V(\beta(x))$, $y$ and~$z$ are adjacent in~$\beta(x)$ if and only if every vertex in~$V_y$ is adjacent in~$G$ to every vertex in~$V_z$, similarly,~$y$ and~$z$ are not adjacent if and only if there is no edge in~$G$ between a vertex in~$V_y$ and a vertex in~$V_z$.
Hence, it is impossible that there are vertex pairs~$(v_1,w_1)\in V_y \times V_z$ and~$(v_2,w_2)\in V_y \times V_z$ such that~$v_1$ is adjacent to~$w_1$ and~$v_2$ is not adjacent to~$w_2$.

We call~$\beta(x)$ the~\emph{quotient graph} of~$x$.
The~\emph{width of a modular decomposition} is the size of a largest vertex set of any quotient graph and the \emph{modular-width} of a graph~$G$, denoted by~$\mw(G)$, is the minimal width of any modular decomposition of~$G$~\cite{GLO13}.

\paragraph{Splitwidth.}
A~\emph{split} of a graph~$G=(V,E)$ is a partition~$(V_1, V_2)$ of~$V$ with~$|V_1| \geq 2$ and~$|V_2|\geq 2$ such that all vertices in~$V_1$ with at least one neighbor in~$V_2$ have the same neighborhood in~$V_2$.
In other words, there are sets~$V_1'\subseteq V_1$ and~$V_2'\subseteq V_2$ such that~$N(v)\cap V_2 = V'_2$ for each~$v\in V'_1$ and~$N(w)\cap V_2 = \emptyset$ for each~$w\in V_1\setminus V'_1$.
If there is no split for~$G$, we call~$G$~\emph{prime}.
Let~$(V_1,V_2)$ be a split of~$G$.
A~\emph{simple decomposition of~$G$ with respect to~$(V_1,V_2)$} consists of two graphs~$G_1$ and~$G_2$ where~$G_i = (W_i, E_i)$ for~$i\in \{1,2\}$ such that~$W_i = V_i \cup \{x\}$ for some vertex~$x$ which is not contained in~$V$, $G_i[V_i] = G[V_i]$ and~$x$ is adjacent to exactly the vertices of~$V'_i$ in~$G_i$.
The vertex~$x$ is called a~\emph{marker} vertex.
Conversely, two graphs~$G_1=(V_1,E_1)$ and~$G_2=(V_2,E_2)$ with~$V_1\cap V_2 = \{x\}$ can be \emph{composed} into a graph~$G \coloneqq  (V,E)$ as follows: 
The graph~$G$ is the union of~$G_1$ and~$G_2$ without the marker vertex~$x$ plus the edges between each neighbor of~$x$ in~$G_1$ and each neighbor of~$x$ in~$G_2$.
Formally, $V \coloneqq  (V_1 \cup V_2) \setminus \{x\}$ and~$E\coloneqq E_{G_1}(V_1\setminus \{x\}) \cup E_{G_2}(V_2\setminus \{x\}) \cup \{\{v_1, v_2\}: v_1 \in N_{G_1}(x), v_2 \in N_{G_2}(x)\}$.
Note that~$(V_1\setminus \{x\}, V_2\setminus\{x\})$ is a split for~$G$ and~$G_1$ and~$G_2$ are a simple decomposition of~$G$ with respect to~$(V_1\setminus \{x\}, V_2\setminus\{x\})$.
A~\emph{split decomposition}~$(\mathcal{T},\beta)$ of a graph~$G$ consists of an undirected tree~$\mathcal{T}=(\mv, \mathcal{E})$ and a function~$\beta$ that maps each of the nodes~$x$ of~$\mv$ to a prime graph~$\beta(x)$ such that
\begin{itemize}
\item $|V(\beta(x)) \cap V(\beta(y))| = 1$ if~$x$ and~$y$ are adjacent in~$\mT$,
\item $V(\beta(x)) \cap V(\beta(y)) = \emptyset$ if~$x$ and~$y$ are nonadjacent in~$\mT$, and
\item $G$ is equivalent to the graph obtained from recursively composing the graphs of all adjacent node pairs in~$\mT$.
\end{itemize} 
The~\emph{width of a split decomposition} is the size of the largest vertex set of any prime graph and the \emph{splitwidth},~denoted by~$\sw(G)$, of a graph~$G$ is the minimal width of any split decomposition of~$G$~\cite{Cun82}.

\subsubsection{Modulator-Based Parameters}\label{sec:subset_params}
A~\emph{graph property~$\Pi$} is a collection of graphs.
Let~$G$ be a graph and let~$\Pi$ be a graph property.
We call a vertex set~$S\subseteq V(G)$ a~\emph{modulator of~$G$ to~$\Pi$} if~$G-S$ is contained in~$\Pi$. 
In the following, we define several~\emph{modulator parameters} we analyze in this work, that is, graph parameters that are defined as the size of a smallest modulator to some fixed graph property~$\Pi$.

\paragraph{Vertex cover number.}
The~\emph{vertex cover number} of a graph~$G$, denoted by~$\vc(G)$, is the size of a smallest modulator of~$G$ to the graph property that consists of all edgeless graphs.

\paragraph{Bounded-degree deletion.}
Let~$r> 0$ be an integer.
The~\emph{$r$-bounded-degree deletion number} of a graph~$G$, denoted by~$\bdd[r](G)$, is the size of a smallest modulator of~$G$ to the graph property that consists of all graphs of maximum degree at most~$r$~\cite{BBNU12,GKO21,KHMN09}.

\paragraph{$d$-path vertex cover number.}
Let~$d> 0$ be an integer.
The~\emph{$d$-path vertex cover number} of a graph~$G$, denoted by~$\pvc[d](G)$, is the size of a smallest modulator of~$G$ to the graph property that consists of all graphs that contain no path of length~$d$~\cite{BKKS11}.

\paragraph{Feedback vertex set number.}
The~\emph{feedback vertex set number} of a graph~$G$, denoted by~$\fvs(G)$, is the size of a smallest modulator of~$G$ to the graph property that consists of all acyclic graphs~\cite{BBF99,IK21,LN22}.

\paragraph{Cluster vertex deletion number.}
The~\emph{cluster vertex deletion number} of a graph~$G$, denoted by~$\cvd(G)$, is the size of a smallest modulator of~$G$ to the cluster graphs which are the graphs where each connected component is a clique~\cite{DK12,HKMN10}.

\paragraph{Distance to cograph.}
The~\emph{distance to cograph} of a graph~$G$, denoted by~$\dco(G)$, is the size of a smallest modulator of~$G$ to the cographs which are defined as the graphs that do not contain a path on four vertices $(P_4)$ as induced subgraph or, equivalently, the graphs with modular-width~2~\cite{Cai03a,NG10}.

\paragraph{Vertex integrity.} For a graph~$H$, let $\cc(H)$ denote the number of vertices in a largest connected component of $H$. The \emph{vertex integrity} of~$G$ is defined as~$\vi(G) \coloneqq \min_{X \subseteq V(G)}(|X|+ \cc(G-X))$~\cite{BES87,DDH16,GHK+24}. Any set~$X$ with~$\vi(G) = |X| + \cc(G - X)$ is called a \emph{vi-set} of~$G$.

\subsubsection{Treewidth and Treedepth}
Finally, we give the definition for treewidth and the related parameter treedepth.
\paragraph{Treewidth.}
A \emph{tree decomposition} of a graph~$G=(V,E)$ is a pair~$(\mT,\beta)$ consisting of a rooted tree~$\mT=(\mv, \mathcal{A})$ and a function~$\beta\colon  \mv \to 2^V$ such that
\begin{enumerate}
\item for each vertex~$v$ of~$V$, there is at least one node~$x\in \mv$ with~$v\in \beta(x)$,
\item for each edge~$\{u,v\}$ of~$E$, there is at least one node~$x\in \mv$ such that~$\beta(x)$ contains~$u$ and~$v$, and
\item for each vertex~$v\in V$, the subgraph~$\mT[\mv_v]$ is connected, where~$\mv_v \coloneqq   \{x\in \mv: v\in \beta(x)\}$.
\end{enumerate}
We call~$\beta(x)$ the~\emph{bag} of~$x$.
The~\emph{width of a tree decomposition} is the size of the largest bag minus one and the \emph{treewidth} of a graph~$G$, denoted by~$\tw(G)$, is the minimal width of any tree decomposition of~$G$~\cite{Halin1976,RS86}.

\paragraph{Treedepth.}
A \emph{treedepth decomposition} of~$G$ is a rooted tree~$T$ using the same vertex set~$V(G)$ such that for each edge~$\{u,w\}\in E(G)$, either~$u$ is an ancestor of~$w$ or~$w$ is an ancestor of~$u$ in~$T$. 
The \emph{depth} of~$T$ is the maximum number of vertices on any root-to-leaf path. 
The \emph{treedepth}~$\tdp(G)$ of~$G$ is the minimum depth among all treedepth decompositions~\cite{NM06}.

\section{Exact Parameter Values}
\label{sec:eval}

\subsection{Experimental Setup}

Our experiments were performed on a server with two Intel(R) Xeon(R) Gold 6526Y CPUs with 2.80 GHz,~$16$~cores, and~$1024$~GB~RAM. Each individual experiment was performed on one core and was allowed to use up to~$28$~GB RAM and run for up to~$24$~hours. Our Python code was executed using Python 3.13.5. To solve the ILPs we used Gurobi 12.0.1~\cite{gurobi}. To compile and run the existing solvers that were implemented in Java we used OpenJDK 21.0.8. To compile the existing solvers that were implemented in C or C++ we used CMake 3.28.3 and GCC 13.3.0.

The source code of all solvers, all graph data, and our experimental results are all publicly available at 
\begin{quote}
\url{https://www.fmi.uni-jena.de/en/19723/parameter-report}.
\end{quote}

More precisely, the source code of all solvers is available at
\begin{quote}
  \url{https://git.uni-jena.de/algo-engineering/param-report}
\end{quote}
while the  graph data and our experimental results are all available at
\begin{quote}
  \url{https://git.uni-jena.de/algo-engineering/data/graph-repo}.
\end{quote}
  For archival purposes, the data set and code corresponding to the current version of this work are also available via Zenodo~\cite{KMSS25}. 

\subsection{Data Set}
We collected real-world graphs from various applications and repositories such as Konect~\cite{Kunegis13}, Matrix Market~\cite{BPR+97}, Network Repository~\cite{RA15}, or SNAP~\cite{snapnets} with a focus on obtaining graphs of moderate size that are indeed undirected (and not just underlying graphs of directed graphs). The graphs arise, for example, in social network analysis, bioinformatics, scientific computing, or route planning. For most of the graphs, a detailed description of the content is given in the above-mentioned data repository.  
 Overall, we collected~$144$ graphs with~$20$ of them being weighted graphs (where we ignored the weights) and~$13$ being ego networks where we took the largest connected component of the graph induced by the open neighborhood of some vertex in a larger graph.\footnote{In social network analysis, the term ego network refers to the graph in the neighborhood of a vertex and it is a very common analysis step to work on ego networks. We chose the largest connected component because most parameters can be calculated independently on connected components, so the largest connected component is more indicative of instance difficulty than the whole graph.}  
 
 We partitioned these graphs into two categories as follows: For 59 graphs we were able to compute all parameter values exactly; we refer to them as the \emph{easy} graphs. The remaining 85 graphs for which at least one parameter value could not be computed exactly are referred to as the \emph{hard} graphs. A complete list of the easy and hard graphs can be found in~\Cref{sec:instances} in \Cref{table-instance-spec-easy} and \Cref{table-instance-spec-hard}, respectively.

\subsection{Single Parameters}
\paragraph{Absolute Parameter Values.} Arguably the most important aspect that we would like to consider is the values of the parameters in our study. Some statistics for these values and also for the relation between these parameters and~$n$ are presented in \Cref{fig:violin-plots,tab:sats2}. 

\begin{figure*}[t!]
        \subfloat[Degree-related parameters]{%
            \includegraphics[width=.49\linewidth]{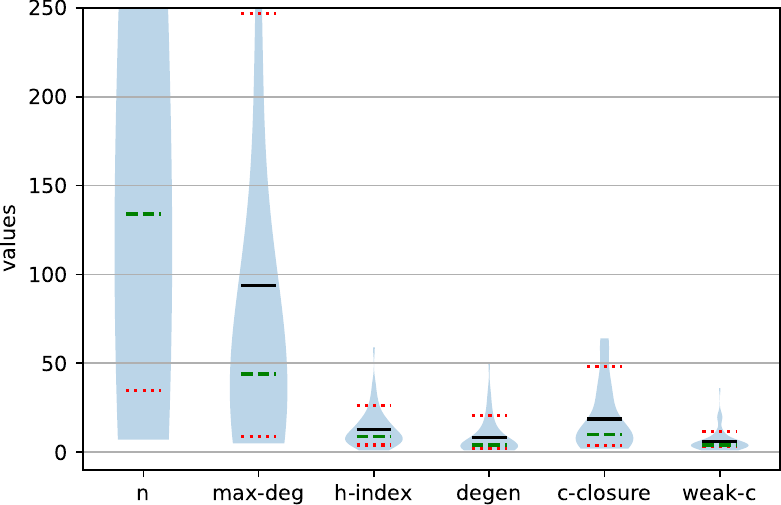}%
            \label{subfig:a}%
        }\hfill
        \subfloat[Neighborhood-based parameters]{%
            \includegraphics[width=.49\linewidth]{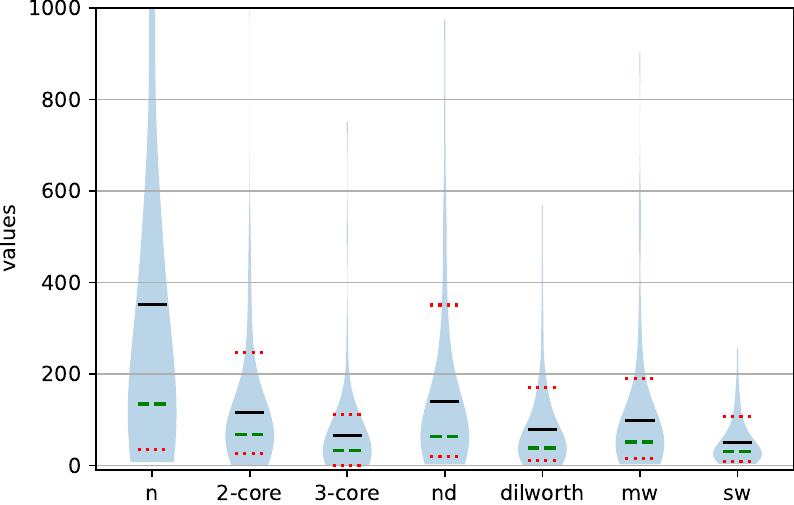}%
            \label{subfig:b}%
        }\\[4ex]
       \centering  
        \subfloat[Modulator- and tree-based parameters]{%
            \includegraphics[width=.49\linewidth]{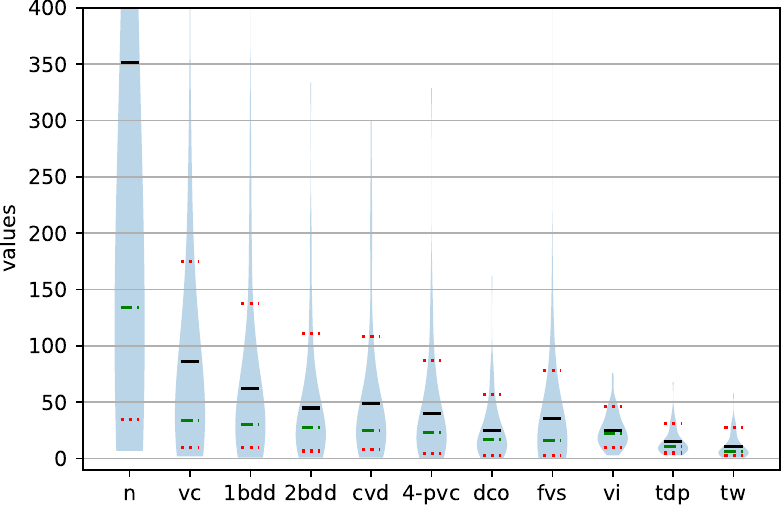}%
            \label{subfig:c}%
        }\hfill
        \caption{Distribution of the absolute parameter values on the 59 easy graphs.
        The thick lines indicate the 10th percentile (dotted red), the average (solid black), the median (dashed green), and the 90th percentile (dotted red), respectively.
        The shape of the violin visualizes the value distribution.
      }
        \label{fig:violin-plots}
    \end{figure*}

\begin{table}[t]
\tiny
\caption{Average, median, and 90th percentile of~$k$ and~$k/n$ for all parameters~$k$.}
\centering
\label{tab:sats2}
\begin{tabular}{l r r r r r r}
\toprule
$k$ & avg. $k$ & median $k$ & 90th p. $k$ & avg. $k/n$ & median $k/n$ & 90th p. $k/n$ \\
\midrule
$n$ & 351.7 & 134.0 & 1\,058.2 & 1.00 & 1.00 & 1.00 \\
\midrule
$\Delta$ & 93.9 & 44.0 & 247.0 & 0.39 & 0.36 & 0.93 \\
$h$-index & 12.8 & 9.0 & 26.2 & 0.13 & 0.08 & 0.32 \\
$d$ & 8.3 & 4.0 & 20.4 & 0.09 & 0.04 & 0.25 \\
$c$ & 18.7 & 10.0 & 48.4 & 0.18 & 0.06 & 0.49 \\
$\gamma$ & 6.1 & 4.0 & 11.4 & 0.07 & 0.03 & 0.20 \\
\midrule
$\core$ & 115.3 & 67.0 & 247.2 & 0.60 & 0.74 & 1.00 \\
$\core[3]$ & 66.0 & 33.0 & 110.8 & 0.44 & 0.51 & 1.00 \\
\midrule
$\nd$ & 139.4 & 63.0 & 350.4 & 0.61 & 0.71 & 1.00 \\
$\nabla$ & 77.9 & 38.0 & 170.6 & 0.37 & 0.33 & 0.71 \\
$\mw$ & 97.4 & 51.0 & 189.6 & 0.51 & 0.56 & 1.00 \\
$\sw$ & 49.4 & 31.0 & 107.2 & 0.40 & 0.17 & 0.99 \\
\midrule
$\vc$ & 85.8 & 34.0 & 175.2 & 0.39 & 0.43 & 0.73 \\
$\bdd$ & 62.3 & 30.0 & 137.6 & 0.31 & 0.32 & 0.59 \\
$\bdd[2]$ & 44.9 & 28.0 & 111.4 & 0.24 & 0.21 & 0.52 \\
$\cvd$ & 48.9 & 25.0 & 108.4 & 0.25 & 0.24 & 0.45 \\
$\pvc$ & 39.7 & 23.0 & 87.4 & 0.24 & 0.23 & 0.54 \\
$\dco$ & 25.0 & 17.0 & 56.6 & 0.16 & 0.11 & 0.34 \\
\midrule
$\fvs$ & 35.5 & 16.0 & 78.0 & 0.22 & 0.19 & 0.55 \\
$\vi$ & 25.1 & 22.0 & 46.2 & 0.23 & 0.15 & 0.53 \\
$\tdp$ & 14.8 & 11.0 & 31.0 & 0.17 & 0.09 & 0.44 \\
$\tw$ & 10.8 & 6.0 & 28.0 & 0.13 & 0.06 & 0.31 \\
\bottomrule
\end{tabular}

\end{table}


We can make the following observations. First, the maximum degree~$\Delta$ has reasonable parameter values in the median, but there are some instances with prohibitively high values. All the degree-based parameters that are upper-bounded by~$\Delta$ have among the best statistics of all considered parameters. Unsurprisingly, degeneracy and weak-$\gamma$-closure have very low parameter values for most instances.

The 2-core is rather large for most of the instances in terms of absolute values and in terms of the relation to~$n$. The 3-core is often substantially smaller than the 2-core so it is motivated to consider this parameter whenever we have positive results for parameterization by the 2-core size.


For the neighborhood-based parameters, the values are rather large when compared to the other parameter groups. The neighborhood diversity \nd is the largest of these parameters and to provide a running time improvement, an FPT-algorithm for~$\nd$ would need to have essentially the same bound as a known algorithm for~$n$. In comparison, the Dilworth number~$\nabla$ and the splitwidth~\sw are more motivated from the perspective of parameter values. The statistics of the parameter values for splitwidth~\sw{} are even better than those for the Dilworth number~$\nabla$; modular-width is better than neighborhood diversity~\nd{} but worse than splitwidth~\sw{} and Dilworth number~$\nabla$.


The modulator-based parameters are typically smaller than the neighborhood-based parameters. 
As maybe expected, the vertex cover number has rather large values. 
In comparison, \bdd is considerably smaller, and \cvd and \bdd[2] are mostly the same and even smaller than \bdd. In turn, the other parameters in that group, \dco and \pvc,  have substantially smaller values than \bdd[2] and~\cvd. Thus, from a data-driven standpoint these parameters should receive more consideration. 
Parameterization by \vc should be considered only as a first step and gives a substantial speed-up only if the running time bound for \vc is roughly the same as for $n$.
Maybe surprisingly, the vertex cover has roughly similar values to the Dilworth number and is overall somewhat smaller than the maximum degree.

The feedback vertex set number~\fvs is of course substantially larger than~\tw on some instances, but it also regularly takes on small values. Thus, one may consider developing FPT-algorithms for~\fvs even when the problem is FPT for treewidth, for example because such algorithms may avoid dynamic programming, which is typical for treewidth algorithms and has several drawbacks including exponential space requirements and a tendency that the worst-case running time is often met in practice.

\paragraph{Klam-Values.}

A classic idea for computing the range of feasible parameter values for an FPT-algorithm is Klam-Values, proposed by Downey and Fellows~\cite{DF99}. For a given FPT-algorithm with running time~$f(k)\cdot \poly(n)$, this is the largest value of~$k$ such that~$f(k)\le 10^{20}$. Surely, this is a simplification of the actual border of tractability: the choice of $10^{20}$ may seem arbitrary, the polynomial running time factor is ignored, and we assume that the actual running times are close to the proven worst-case running time bound. Nevertheless, this value serves as a rough estimate of the range of tractable parameters. 

For a variety of typical running times~$f(k)$, we count how many of the $59$ easy instances would be solvable for the observed parameter values. The results are shown in \Cref{table-para-statistic-klam}. One can see for example, that a slightly subexponential running time~$2^{k/\log k}$ would give feasible algorithms for almost all instance-parameter combinations. For a running time of~$2^k$, still at least half of the instances are solvable for the majority of the parameters. For a running time of~$k^k$, one needs to consider either the smaller degree-based parameters or \tw to solve half of the instances. Maybe somewhat surprisingly, more than half of the instances can be solved when the parameter is the treewidth and the running time is doubly exponential. 

The table also reveals a danger when comparing two parameters based on the relation to~$n$ as done above: for splitwidth, the average value of~$\sw/n$ is~$0.40$, whereas for the Dilworth number~$\nabla$ the average value is~$0.37$. The situation is similar for the median and 90th percentile of~$k/n$. Given these values, one would prefer parameterization by the Dilworth number. However,
 \Cref{table-para-statistic-klam} shows that splitwidth will lead to more solved instances for every single one of the considered running time bounds. This is due to the fact that statistics on the distribution of $k/n$ may be influenced by many instances where~$k$ and~$n$ are so large that they are clearly outside of the realm of tractable instances. In other words, smaller values of $k/n$ may indicate that we get improved running times, but this improvement is possibly made for instances that are too hard to be solved exactly by FPT-algorithms for the parameter~$k$.

For the \textsc{Vertex Cover} problem, this table can also be used to decide whether the running time bound of~$\Oh(1.26^\vc+n)$ or the bound of~$\Oh(2^\tw\cdot n)$ is better on this data: We see that~$2^\tw$ is below~$10^{20}$ for all instances while this is barely not the case for~$2^{\vc/ \log \vc}$ which is smaller than~$1.26^\vc$ for the observed parameter values.\footnote{The advantage of~$\tw$ persists also when directly comparing the precise running time bounds including the dependence on~$n$: the bound $2^\tw\cdot n$ is below~$10^{20}$ for all $59$ instances while~$1.26^{\vc}+n$ is below ~$10^{20}$ for $53$ instances.} Hence, under the stated assumptions, the algorithm for treewidth would be preferable.    
 \begin{table}[t]
\tiny
\caption{Number of benchmark instances with~$k$ being below the Klam-Values for different running time bounds~$f(k)$. The total number of instances is 59.}
\label{table-para-statistic-klam}
\centering
\begin{tabular}{l r r r r r r r r}
\toprule
parameter & $2^{k/\log(k)}$ & $\sqrt{2}^k$ & $2^k$ & $4^k$ & $k!$ & $k^k$ & $2^{k^2}$ & $2^{(2^k)}$ \\
\midrule
$\Delta$ & 59 & 47 & 38 & 27 & 19 & 14 & 6 & 4 \\
$h$-index & 59 & 59 & 59 & 56 & 51 & 44 & 27 & 19 \\
$d$ & 59 & 59 & 59 & 57 & 53 & 50 & 39 & 39 \\
$c$ & 59 & 59 & 59 & 48 & 41 & 37 & 24 & 19 \\
$\gamma$ & 59 & 59 & 59 & 58 & 58 & 54 & 48 & 46 \\
\midrule
$\core$ & 58 & 49 & 29 & 11 & 4 & 4 & 2 & 1 \\
$\core[3]$ & 58 & 54 & 42 & 31 & 25 & 21 & 17 & 15 \\
\midrule
$\nd$ & 57 & 42 & 30 & 13 & 7 & 5 & 3 & 3 \\
$\nabla$ & 59 & 49 & 39 & 25 & 15 & 14 & 3 & 3 \\
$\mw$ & 58 & 49 & 36 & 20 & 13 & 8 & 4 & 3 \\
$\sw$ & 59 & 54 & 45 & 31 & 20 & 15 & 8 & 5 \\
\midrule
$\vc$ & 58 & 48 & 39 & 29 & 17 & 13 & 5 & 4 \\
$\bdd$ & 59 & 53 & 42 & 33 & 21 & 17 & 6 & 5 \\
$\bdd[2]$ & 59 & 55 & 48 & 37 & 28 & 19 & 9 & 6 \\
$\cvd$ & 59 & 54 & 48 & 36 & 26 & 18 & 7 & 5 \\
$\pvc$ & 59 & 56 & 49 & 39 & 27 & 21 & 11 & 9 \\
$\dco$ & 59 & 58 & 54 & 44 & 36 & 29 & 16 & 11 \\
\midrule
$\fvs$ & 59 & 57 & 49 & 43 & 37 & 31 & 20 & 15 \\
$\vi$ & 59 & 59 & 58 & 45 & 29 & 19 & 6 & 4 \\
$\tdp$ & 59 & 59 & 58 & 55 & 48 & 41 & 21 & 10 \\
$\tw$ & 59 & 59 & 59 & 57 & 49 & 47 & 36 & 32 \\
\bottomrule
\end{tabular}

\end{table}

\paragraph{Annotating Parameter Hierarchies.}
\def\nddilmed{0.95}
\def\nddilnp{1.00}

\def\ndmwmed{0.58}
\def\ndmwnp{0.75}

\def\mwswmed{0.82}
\def\mwswnp{1.00}

\def\maxdegcmed{0.42}
\def\maxdegcnp{0.77}

\def\maxdeghidmed{0.36}
\def\maxdeghidnp{0.76}

\def\cgammamed{0.50}
\def\cgammanp{0.80}

\def\hiddegmed{0.57}
\def\hiddegnp{0.95}

\def\deggammamed{1.00}
\def\deggammanp{1.50}

\def\vchidmed{0.28}
\def\vchidnp{0.69}

\def\vcbddmed{0.80}
\def\vcbddnp{1.00}

\def\bddcvdmed{1.00}
\def\bddcvdnp{1.00}

\def\bddtbddmed{0.78}
\def\bddtbddnp{0.98}

\def\bddpvcmed{0.73}
\def\bddpvcnp{0.91}

\def\bddfvsmed{0.58}
\def\bddfvsnp{0.95}

\def\bddvimed{0.80}
\def\bddvinp{1.11}

\def\tbddtwmed{0.35}
\def\tbddtwnp{0.97}

\def\cvddcomed{0.60}
\def\cvddconp{0.81}

\def\pvcdcomed{0.85}
\def\pvcdconp{1.00}

\def\pvctdpmed{0.70}
\def\pvctdpnp{1.25}

\def\pvctwmed{0.44}
\def\pvctwnp{1.00}

\def\fvstwmed{0.62}
\def\fvstwnp{1.00}

\def\vitdpmed{0.63}
\def\vitdpnp{1.00}

\def\tdptwmed{0.64}
\def\tdptwnp{0.90}

\def\twdegmed{0.75}
\def\twdegnp{1.00}

\def\twocorefvsmed{0.25}
\def\twocorefvsnp{0.64}

\def\twocorethreecoremed{0.64}
\def\twocorethreecorenp{1.00}


\usetikzlibrary{shapes}

\tikzstyle{para}=[rectangle,draw=black,minimum height=.8cm,minimum width=1.3cm,fill=gray!10,rounded corners=1mm, on grid]

\newcommand{\tworows}[2]{\begin{tabular}{c}{#1}\\{#2}\end{tabular}}
\newcommand{\distto}[1]{\tworows{Distance to}{#1}}

\definecolor{r3}{rgb}{1, 0.4, 0.3}
\definecolor{g}{rgb}{0.2, 0.9, 0.3}
\definecolor{y2}{rgb}{1, 1, 0.3}

\newcommand{\centerBox}[2]{\begin{tabular}{c}
$#2$\\ 
$#1$
\end{tabular}}

\begin{figure}[t]
\centering
\small
\begin{tikzpicture}[node distance=2*0.75cm and 3.7*0.38cm, every node/.style={scale=0.57}]
\linespread{1}
\node[para,fill=g,xshift=1.5cm] (vc) {\vc};


\node[para,fill=g,xshift=7cm] (twocore)  {\core[2]};

\node[para,fill=g,xshift=1cm] (threecore) [below= of twocore] {\core[3]};
\draw (twocore) edge[bend left=20] node [midway, fill=white, inner sep = 0pt] {\centerBox{\twocorethreecorenp}{\twocorethreecoremed}}(threecore) ;

\node[para,fill=g,xshift=2cm] (fvs) [below left= of threecore] {\fvs};

\node[para,fill=g] (bdd) [below= of vc] {\bdd};
\draw (vc) -- (bdd) node [midway, fill=white, inner sep = 0pt] {\centerBox{\vcbddnp}{\vcbddmed}} ;

\node[para,fill=g,xshift=-3.5cm] (dcl) [below= of bdd] {\cvd};
\draw (bdd) edge[bend right=10] node [midway, fill=white, inner sep = 0pt] {\centerBox{\bddcvdnp}{\bddcvdmed}} (dcl) ;

\node[para,fill=g,xshift=-0.1cm] (tbdd) [left= of fvs] {\bdd[2]};
\draw (bdd) edge[bend left=10] node [midway, fill=white, inner sep = 0pt] {\centerBox{\bddtbddnp}{\bddtbddmed}}  (tbdd);

\node[para,fill=g,xshift=-1cm] (fourpvc) [below= of bdd] {\pvc};
\draw (bdd) edge[bend right=10]  node [midway, fill=white, inner sep = 0pt] {\centerBox{\bddpvcnp}{\bddpvcmed}} (fourpvc);

\node[para,fill=g,xshift=0.1cm] (vi) [left= of tbdd] {\vi};
\draw (bdd) edge node [midway, fill=white, inner sep = 0pt] {\centerBox{\bddvinp}{\bddvimed}} (vi) ;

\node[para,fill=g, xshift=-1.0cm] (tdp) [below= of vi] {\tdp};
\draw (vi) edge [bend left=10] node [midway, fill=white, inner sep = 0pt] {\centerBox{\vitdpnp}{\vitdpmed}} (tdp) ;

\node[para,fill=g,xshift=-10cm] (nd) {\nd};

\node[para,fill=g] (dwn) [below left= of nd] {$\nabla$};
\draw (nd) edge[bend right=20] node [midway, fill=white, inner sep = 0pt] {\centerBox{\nddilnp}{\nddilmed}} (dwn) ;

\node[para,fill=g,,xshift=-1.75cm] (mdw) [below right= of nd] {\mw};
\draw (nd) -- (mdw) node [midway, fill=white, inner sep = 0pt] {\centerBox{\ndmwnp}{\ndmwmed}} ;
\node[para,fill=g] (spw) [below=  of mdw] {\sw};
\draw (mdw) -- (spw) node [midway, fill=white, inner sep = 0pt] {\centerBox{\mwswnp}{\mwswmed}} ;

\node[para,fill=g,xshift=0.5cm] (distco) [below= of dcl] {\dco};
\draw (dcl) edge[bend right=10] node [midway, fill=white, inner sep = 0pt] {\centerBox{\cvddconp}{\cvddcomed}}(distco);
\draw (fourpvc) -- (distco) node [midway, fill=white, inner sep = 0pt] {\centerBox{\pvcdconp}{\pvcdcomed}} ;


\draw (bdd) edge[bend left=10] node [midway, fill=white, inner sep = 0pt] {\centerBox{\bddfvsnp}{\bddfvsmed}} (fvs);
\draw (twocore) edge[bend right=30] node [midway, fill=white, inner sep = 0pt] {\centerBox{\twocorefvsnp}{\twocorefvsmed}}  (fvs);

\node[para, xshift=1cm,fill=g] (tw) [below= of tdp]  {\tw};
\draw (tbdd) edge [bend left=20]  node [midway, fill=white, inner sep = 0pt] {\centerBox{\tbddtwnp}{\tbddtwmed}} (tw);
\draw (fourpvc) edge [bend right=20] node [midway, fill=white, inner sep = 0pt] {\centerBox{\pvctdpnp}{\pvctdpmed}} (tdp) ;
\draw (tdp) edge node [midway, fill=white, inner sep = 0pt] {\centerBox{\tdptwnp}{\tdptwmed}} (tw) ;

\draw (fvs) edge[bend left=20]  node [midway, fill=white, inner sep = 0pt] {\centerBox{\fvstwnp}{\fvstwmed}} (tw);

\node[para,fill=g,xshift=-5cm,] (mxd)  {$\Delta$};

\node[para,fill=g,xshift=0.2cm] (ccl) [below left= of mxd] {$c$};
\draw [bend right=20](mxd) -- (ccl) node [midway, fill=white, inner sep = 0pt] {\centerBox{\maxdegcnp}{\maxdegcmed}};
\node[para,fill=g,xshift=3cm] (hid) [below= of ccl] {$h$-index};
\draw (vc) -- (hid) node [midway, fill=white, inner sep = 0pt] {\centerBox{\vchidnp}{\vchidmed}} ;
\draw (mxd) --  (hid) node [midway, fill=white, inner sep = 0pt] {\centerBox{\maxdeghidnp}{\maxdeghidmed}} ;

\node[para,fill=g,xshift=1.5cm,yshift=-4.0cm] (deg) [below= of hid] {$d$};
\draw (hid) -- (deg) node [midway, fill=white, inner sep = 0pt] {\centerBox{\hiddegnp}{\hiddegmed}} ;
\draw  (tw) edge  node [midway, fill=white, inner sep = 0pt] {\centerBox{\twdegnp}{\twdegmed}} (deg);

\node[para,fill=g] (gcl) [below = of deg] {$\gamma$};
\draw (ccl) edge[bend right=20] node [midway, fill=white, inner sep = 0pt] {\centerBox{\cgammanp}{\cgammamed}} (gcl);
\draw (deg) -- (gcl) node [midway, fill=white, inner sep = 0pt] {\centerBox{\deggammanp}{\deggammamed}} ;

\end{tikzpicture}
\caption{An annotated parameter hierarchy. For an edge from parameter~$k$ above to parameter~$\ell$ below, the upper label of an edge is the median value of~$\ell/k$, the lower label of an edge is the 90th percentile value of~$\ell/k$. 
Note that the 90th percentile of~$\vi$ is larger than that of~$\bdd$ since in the former we additionally add the component size.}
	\label{fig-results}
\end{figure}


A common strategy for exploring the parameterization space is to use diagrams that represent parameter relations as a guideline~\cite{Jan13,SW19,Tran22}\footnote{See also the websites \url{https://manyu.pro/assets/parameter-hierarchy.html} and \url{https://vaclavblazej.github.io/parameters/}.}. In this diagram, parameters with large values are located in the upper part, whereas parameters with small values are located in the lower part. An edge between two parameters represents a relation between them. Roughly speaking this means that the lower parameter is guaranteed to be smaller than the higher parameter. The concrete definition of what it means to be smaller may differ from diagram to diagram. Here, we consider a very strict definition, where an edge from a parameter~$k$ to a parameter~$\ell$ below means that in every graph~$G$, we have~$k(G)\ge\ell(G) +\Oh(1)$. The additive constant usually only concerns offsets that are created by design choices in the parameter definition. For example, for the parameters degeneracy~$d$ and weak closure $\gamma$, we have~$d\le \gamma+1$ but~$d\le \gamma$ is not always true.  

When pondering which parameters to consider for a particular problem, the process is often to start with large parameters in the hierarchy and then to proceed downwards as long as one still obtains FPT-algorithms. There may be, however, multiple options for proceeding downwards, as can be seen in \Cref{fig-results}. The choice of the next parameterization may also benefit from a data-driven approach: We should try to get FPT-results for the parameter whose values give the biggest improvement relative to the values of the current parameter, for which we have shown an FPT-algorithm. To provide an estimate of this improvement, we propose to annotate the parameter hierarchy by statistics on the relation between the parameters. The annotations for the considered parameters are shown in \Cref{fig-results}. Note that for the weak closure~$\gamma$ and the degeneracy~$d$, the value of the 90th percentile of~$\gamma/d$ is larger than~1. This is due to the above-mentioned additive constant in the bound~$d\le \gamma+1$: 10 of the instances have degeneracy 2 and weak-closure~3.

An example scenario, where the annotations may be helpful is as follows. Assume that we have shown an FPT-algorithm for~\bdd. Then, according to \Cref{fig-results}, one should consider \fvs or \pvc next since in our set of graphs they provide the biggest improvement.
In contrast, \cvd is not such a good choice since for half of the instances we essentially have~$\cvd(G)=\bdd(G)$.

\subsection{Multiple Parameters}
We also consider the influence of combining noncomparable parameters in different ways.

\def\avgcvd{48.9}
\def\medcvd{25.0}
\def\percvd{108.4}
\def\avgbdd{44.9}
\def\medbdd{28.0}
\def\perbdd{111.4}
\def\avgpvc{39.7}
\def\medpvc{23.0}
\def\perpvc{87.4}
\def\avgdco{25.0}
\def\meddco{17.0}
\def\perdco{56.6}
\def\avgfvs{35.5}
\def\medfvs{16.0}
\def\perfvs{78.0}
\def\avgcvdbdd{\textbf{38.7}}
\def\medcvdbdd{\textbf{20.0}}
\def\percvdbdd{\textbf{108.2}}
\def\avgcvdpvc{\textbf{33.5}}
\def\medcvdpvc{\textbf{20.0}}
\def\percvdpvc{\textbf{83.4}}
\def\avgbddpvc{\textbf{38.6}}
\def\medbddpvc{\textbf{20.0}}
\def\perbddpvc{87.4}
\def\avgcvdfvs{\textbf{25.3}}
\def\medcvdfvs{\textbf{15.0}}
\def\percvdfvs{\textbf{52.0}}
\def\avgbdddco{\textbf{24.9}}
\def\medbdddco{17.0}
\def\perbdddco{56.6}
\def\avgbddfvs{\textbf{31.2}}
\def\medbddfvs{\textbf{15.0}}
\def\perbddfvs{\textbf{67.8}}
\def\avgpvcfvs{\textbf{31.3}}
\def\medpvcfvs{\textbf{15.0}}
\def\perpvcfvs{\textbf{68.6}}
\def\avgdcofvs{\textbf{18.5}}
\def\meddcofvs{\textbf{11.0}}
\def\perdcofvs{\textbf{42.2}}

\usetikzlibrary{shapes}

\tikzstyle{para}=[rectangle,draw=black,minimum height=.8cm,fill=gray!10,rounded corners=1mm, on grid]

\begin{figure}[t]
\centering
\small
\begin{tikzpicture}[node distance=2*0.45cm and 3.7*0.38cm, every node/.style={scale=0.57}]
\linespread{1}

\node[para,fill=g,xshift=1.5cm, text width=1.0cm,align=center, yshift=0.9cm] (cvd) {\cvd \\ {\avgcvd} \\ {\medcvd} \\ {\percvd}};

\node[para,fill=g,xshift=6cm, text width=1.0cm,align=center, yshift=0.9cm] (bdd)  {\bdd[2] \\ {\avgbdd} \\ {\medbdd} \\ {\perbdd}};

\node[para,fill=g,xshift=10.5cm, text width=1.0cm,align=center, yshift=0.9cm] (pvc)  {\pvc \\ {\avgpvc} \\ {\medpvc} \\ {\perpvc}};

\node[para,fill=g,xshift=14.5cm, text width=1.0cm,align=center, yshift=0.9cm] (dco)  {\dco \\ {\avgdco} \\ {\meddco} \\ {\perdco}};

\node[para,fill=g,xshift=19cm, text width=1.0cm,align=center, yshift=0.9cm] (fvs)  {\fvs \\ {\avgfvs} \\ {\medfvs} \\ {\perfvs}};

\node[para,fill=g,xshift=0cm,yshift=-4cm, text width=1.9cm,align=center] (cvdbdd) {$\cvd + \bdd[2]$\\ {\avgcvdbdd} \\ {\medcvdbdd} \\ {\percvdbdd}};
\draw (cvd) to (cvdbdd);
\draw (bdd) to (cvdbdd);

\node[para,fill=g,xshift=3cm,yshift=-4cm, text width=1.9cm,align=center] (cvdpvc) {$\cvd + \pvc$\\ {\avgcvdpvc} \\ {\medcvdpvc} \\ {\percvdpvc}};
\draw (cvd) to (cvdpvc);
\draw (pvc) to (cvdpvc);

\node[para,fill=g,xshift=6cm,yshift=-4cm, text width=1.9cm,align=center] (bddpvc) {$\bdd[2] + \pvc$\\ {\avgbddpvc} \\ {\medbddpvc} \\ {\perbddpvc}};
\draw (bdd) to (bddpvc);
\draw (pvc) to (bddpvc);

\node[para,fill=g,xshift=9cm,yshift=-4cm, text width=1.9cm,align=center] (cvdfvs) {$\cvd + \fvs$\\ {\avgcvdfvs} \\ {\medcvdfvs} \\ {\percvdfvs}};
\draw (cvd) to (cvdfvs);
\draw (fvs) to (cvdfvs);

\node[para,fill=g,xshift=12cm,yshift=-4cm, text width=1.9cm,align=center] (bdddco) {$\bdd[2] + \dco$\\ {\avgbdddco} \\ {\medbdddco} \\ {\perbdddco}};
\draw (bdd) to (bdddco);
\draw (dco) to (bdddco);

\node[para,fill=g,xshift=15cm,yshift=-4cm, text width=1.9cm,align=center] (bddfvs) {$\bdd[2] + \fvs$\\ {\avgbddfvs} \\ {\medbddfvs} \\ {\perbddfvs}};
\draw (bdd) to (bddfvs);
\draw (fvs) to (bddfvs);

\node[para,fill=g,xshift=18cm,yshift=-4cm, text width=1.9cm,align=center] (pvcfvs) {$\pvc + \fvs$\\ {\avgpvcfvs} \\ {\medpvcfvs} \\ {\perpvcfvs}};
\draw (pvc) to (pvcfvs);
\draw (fvs) to (pvcfvs);

\node[para,fill=g,xshift=21cm,yshift=-4cm, text width=1.9cm,align=center] (dcofvs) {$\dco + \fvs$\\ {\avgdcofvs} \\ {\meddcofvs} \\ {\perdcofvs}};
\draw (dco) to (dcofvs);
\draw (fvs) to (dcofvs);

\end{tikzpicture}
\caption{An annotated parameter hierarchy for the pairwise minimum of \cvd, \bdd[2], \pvc, \dco, and~\fvs. 
For a parameter~$k$ the three values in its box are the average, median, and 90th percentile, respectively.
A value for a combined parameter is bold if it is smaller than the minimum of the respective values for the two individual parameters.}
	\label{fig-best-of}
\end{figure}


\paragraph{Best-of-Parameterization.} A first implicit type of combination is to take the minimum of two parameters. The scenario for this combination is that whenever we have two different FPT-algorithms, one for the parameter~$k$ and another one for the parameter~$\ell$, we could run for each instance the one whose parameter value is smaller. In other words, we have an FPT-algorithm for the minimum of the two parameters. Now in the algorithm design process, we would like to find parameter combinations where this minimum parameter is particularly worthwhile.

To give an example of such combinations, we consider the pairwise minimum of \cvd, \dco, \bdd[2], \pvc, and~\fvs.
\Cref{fig-best-of} shows the results.
For example, the minimum combination of \cvd with any of the three unrelated parameters in the diagram is notably smaller than either single parameter. Hence, the respective parameters complement each other nicely, with \cvd being small on some instances where the other parameters are large.
In contrast, combining \bdd[2] and \dco only gives a small improvement on average, not on the median or 90th percentile. So either both parameters take on similar values in both instances, or \fvs is in most cases better than~\bdd[2]. If we instead combine \dco with \fvs, then we obtain the overall best parameter combination.

\paragraph{Combined Parameters.}

\begin{table}[t]
\tiny
\caption{Parameter combination of $\tw$ with an unrelated parameter~$k$. In the table we report statistics on the maximum of~$\tw$ and~$k$.}
\label{table-max-para-comb}
\centering

\begin{tabular}{l r r r r r r}
\toprule
 & \multicolumn{2}{c}{avg.} & \multicolumn{2}{c}{median} & \multicolumn{2}{c}{90th p.} \\[1ex]
$k$ & \multicolumn{1}{c}{$\max(k,\tw)$} & \multicolumn{1}{c}{$k$} & \multicolumn{1}{c}{$\max(k,\tw)$} & \multicolumn{1}{c}{$k$} & \multicolumn{1}{c}{$\max(k,\tw)$} & \multicolumn{1}{c}{$k$} \\
\midrule
$\tw$ & 10.8 & 10.8 & 6.0 & 6.0 & 28.0 & 28.0 \\
$\Delta$ & 93.9 & 93.9 & 44.0 & 44.0 & 247.0 & 247.0 \\
$h$-index & 13.2 & 12.8 & 9.0 & 9.0 & 28.0 & 26.2 \\
$c$ & 19.1 & 18.7 & 10.0 & 10.0 & 48.4 & 48.4 \\
$\core$ & 115.3 & 115.3 & 67.0 & 67.0 & 247.2 & 247.2 \\
$\core[3]$ & 66.5 & 66.0 & 33.0 & 33.0 & 110.8 & 110.8 \\
$\nabla$ & 79.0 & 77.9 & 38.0 & 38.0 & 170.6 & 170.6 \\
$\cvd$ & 49.7 & 48.9 & 28.0 & 25.0 & 108.4 & 108.4 \\
\bottomrule
\end{tabular}

\end{table}

A common phenomenon in parameterized complexity is that if a problem is W-hard with respect to some parameter, then adding a further parameter leads to an FPT-algorithm.
Here, we consider the combination of treewidth~\tw with one unrelated parameter~$k$.
For each such parameter, we compute~$\max(\tw(G), k(G))$ on all graphs.
The results are shown in \Cref{table-max-para-comb}.
One can see that the value of~$\max(k,\tw)$ is almost identical to~$k$.
This implies that although \tw is unrelated to parameters such as~$\Delta$ or~$\nabla$, in practice \tw is usually smaller than these parameters.
Additionally, when designing an FPT-algorithm for one of these parameters, one can utilize the additional structure of a tree decomposition without significantly increasing the combined parameter value.

\subsection{Parameterizations Based on Quotient Graphs}
As observed in the analysis of the single parameters, the neighborhood-based parameterizations neighborhood diversity and modular-width are often not small enough to lead to practical FPT running times. Still, for many problems two vertices with similar neighborhoods are, informally speaking, redundant and we would like to have parameters that are small whenever we have many such redundancies. We propose a systematic way to define such parameterizations by basing them on the quotient graphs of the modular decomposition or on the \emph{twin-quotient graph}. 

\paragraph{Twin-Quotient Graph-Based Parameters.} Recall that in the definition of neighborhood diversity, two vertices~$u$ and~$v$ are twins if~$N(u)\setminus \{v\}=N(v)\setminus \{u\}$. The twin relation is an equivalence relation. For a vertex~$u$, let~$[u]$ denote the equivalence class (also called \emph{twin class}) of~$u$.
\begin{definition}
  The \emph{twin-quotient graph} of a graph~$G$ is the graph~$G_\nd$ with vertex set~$V(G_\nd)\coloneqq \{[u]: u\in V(G)\}$ and~$E(G_\nd)\coloneqq \{\{[u],[v]\} : [u]\neq [v] \land \{u,v\}\in E(G) \}$. 
\end{definition}
Informally, the twin-quotient graph has one vertex of every twin class and two such vertices~$t$ and~$t'$ are adjacent if the members of~$t$ are adjacent to the members of~$t'$.  Note that the twin-quotient graph of~$G$ is an induced subgraph of~$G$. We use the notation~$G_\nd$ to highlight that the neighborhood diversity is essentially the number of vertices of that graph. Now for any parameter~$k$ we define a corresponding parameter~$k_\nd$ which is defined just as~$k$ but on the twin-quotient graph. This is formalized as follows.

\begin{definition}Let~$k:\mathcal{G}\to \mathds{N}$ be a parameter. Then, the parameter~$k_\nd:\mathcal{G}\to \mathds{N}$ is defined via $k_\nd(G)\coloneqq  k(G_\nd)$.
\end{definition}
In general, it is not guaranteed that~$k_\nd(G)\le k(G)$ holds for every graph. Most common parameters and all parameters considered in this work, however, behave monotonically with respect to taking induced subgraphs. More precisely, they do not increase whenever we delete a vertex from~$G$. For such parameters,~$k_\nd(G)\le k(G)$ holds since~$G_\nd$ is an induced subgraph of~$G$. This parameterization has been considered for example by~Lampis~\cite{Lampis20} who studied~$\tw_{\nd}$ for vertex coloring.

\paragraph{Modular-Width-Based Parameters.}
For the modular-width, the parameter definition is more involved. Recall that the modular-width is defined over modular decompositions. These are pairs~$(\mT,\beta)$ consisting of a rooted tree~$\mT=(\mv, \mathcal{A})$ and a function~$\beta$ that maps each node~$x\in \mv$ to its quotient graph~$\beta(x)$.  The modular-width is the maximum number of vertices of any~$\beta(x)$. The idea now is simply to consider any parameter~$k$ instead of the number of vertices. Formally, this leads to the following.

\begin{definition}
  Let~$k:\mathcal{G}\to \mathds{N}$ be
  a parameter. 
  Moreover, for a given graph~$G$ and a modular decomposition~$(\mT,\beta)$ of~$G$, we define~$k_\mw((\mT,\beta))\coloneqq\max_{x\in V(\mT)} k(\beta(x))$.
  Then, the parameter~$k_\mw:\mathcal{G}\to \mathds{N}$ is defined as $k_\mw(G)\coloneqq  k_\mw((\mT_G,\beta_G))$, where~$(\mT_G,\beta_G)$ denotes a modular
  decomposition of~$G$ that minimizes~$k_\mw((\mT_G,\beta_G))$. 
\end{definition}

Such a  modular decomposition-based parameter was first considered by Bodlaender and Jansen~\cite{BJ00}, who studied~$\tw_\md$ and also more recently by Hegerfeld and Kratsch~\cite{HK23}. Another application was given in the context of parameterized local search, where an FPT-algorithm for~$\Delta$ was improved to one for~$\Delta_\mw$~\cite{KM22}.

Note that with these definitions, we have~$n_\nd(G)=\nd(G)$ and~$n_\mw(G)=\mw(G)$.  Since we consider parameters~$k$ that are never larger than~$n$, each parameter~$k_\nd(G)$ is simultaneously upper-bounded by~$k$ \emph{and} by~$\nd(G)$. 
Moreover, as the twin-quotient graph of~$G$ is a supergraph of every prime graph in every modular decomposition of~$G$, we further obtain that~$k_\mw\le k_\nd$.

We now describe how we find the right modular decomposition of a graph, that is, the modular decomposition that defines the parameter~$k_\mw$, efficiently.

\begin{lemma}
Let~$k:\mathcal{G}\to \mathds{N}$ be a parameter for which the value does not increase by taking induced subgraphs.
Then, for each graph~$G$, we can compute a modular decomposition~$(\mT_G,\beta_G)$ of~$G$ in linear time, such that~$k_\mw((\mT_G,\beta_G)) = k_\mw(G)$. 
\end{lemma}
\begin{proof}
We can compute in linear time~\cite{DBLP:conf/icalp/TedderCHP08} a modular decomposition~$(\mT,\beta)$ of~$G$ for which each quotient graph~$\beta(x)$ on at least two vertices is~\emph{prime}, that is, the modular width of each quotient graph~$\beta(x)$ equals the number of vertices of~$\beta(x)$.
Let~$(\mT',\beta')$ be any other modular decomposition of~$G$, then for each~$x \in V(\mT)$, there is a~$y\in V(\mT')$, such that~$\beta(x)$ is isomorphic to an induced subgraph of~$\beta'(y)$.
In particular for~$x^* = \arg \min_{x\in V(\mT)} k(\beta(x))$, there is some~$y\in V(\mT')$ where~$\beta(x^*)$ is isomorphic to an induced subgraph of~$\beta'(y)$.
As parameter~$k$ does not increase by taking induced subgraphs, this implies that~$k(\beta(x^*))\leq k(\beta'(y))$, which implies that~$k_\mw((\mT,\beta)) \leq k_\mw((\mT',\beta'))$.
Hence, $k_\mw((\mT,\beta)) = k_\mw(G)$.
\end{proof}

\begin{table}[t]
\tiny
\caption{Values of the parameterization based on twin-quotient graphs,~$k_\nd$, and modular decompositions,~$k_\mw$, relative to the original parameter~$k$.}
\label{table-md-q-easy}
\centering

\begin{tabular}{l r r r r r r}
\toprule
 & \multicolumn{2}{c}{avg.} & \multicolumn{2}{c}{median} & \multicolumn{2}{c}{90th p.} \\[1ex]
 & \multicolumn{1}{c}{$k_{\nd} / k$}& \multicolumn{1}{c}{$k_{\mw} / k$} & \multicolumn{1}{c}{$k_{\nd} / k$} & \multicolumn{1}{c}{$k_{\mw} / k$} & \multicolumn{1}{c}{$k_{\nd} / k$} & \multicolumn{1}{c}{$k_{\mw} / k$} \\
\midrule
$n$ & 0.61 & 0.51 & 0.71 & 0.56 & 1.00 & 1.00 \\
\midrule
$\Delta$ & 0.60 & 0.59 & 0.71 & 0.71 & 1.00 & 1.00 \\
$h$-index & 0.82 & 0.80 & 0.88 & 0.88 & 1.00 & 1.00 \\
$d$ & 0.88 & 0.87 & 1.00 & 1.00 & 1.00 & 1.00 \\
$c$ & 0.73 & 0.72 & 0.87 & 0.87 & 1.00 & 1.00 \\
$\gamma$ & 0.96 & 0.96 & 1.00 & 1.00 & 1.00 & 1.00 \\
\midrule
$\core$ & 0.69 & 0.62 & 0.71 & 0.64 & 1.00 & 1.00 \\
$\core[3]$ & 0.56 & 0.53 & 0.70 & 0.58 & 1.00 & 1.00 \\
\midrule
$\nd$ & 0.90 & 0.75 & 1.00 & 0.94 & 1.00 & 1.00 \\
$\nabla$ & 0.89 & 0.75 & 1.00 & 0.96 & 1.00 & 1.00 \\
$\mw$ & 1.00 & 1.00 & 1.00 & 1.00 & 1.00 & 1.00 \\
$\sw$ & 0.99 & 0.98 & 1.00 & 1.00 & 1.00 & 1.00 \\
\midrule
$\vc$ & 0.83 & 0.69 & 0.99 & 0.91 & 1.00 & 1.00 \\
$\bdd$ & 0.74 & 0.66 & 0.83 & 0.75 & 1.00 & 1.00 \\
$\bdd[2]$ & 0.71 & 0.64 & 0.79 & 0.75 & 1.00 & 1.00 \\
$\cvd$ & 0.81 & 0.70 & 0.93 & 0.91 & 1.00 & 1.00 \\
$\pvc$ & 0.86 & 0.74 & 1.00 & 0.98 & 1.00 & 1.00 \\
$\dco$ & 0.94 & 0.81 & 1.00 & 1.00 & 1.00 & 1.00 \\
\midrule
$\fvs$ & 0.77 & 0.70 & 0.88 & 0.86 & 1.00 & 1.00 \\
$\vi$ & 0.85 & 0.81 & 0.92 & 0.89 & 1.00 & 1.00 \\
$\tdp$ & 0.93 & 0.92 & 1.00 & 1.00 & 1.00 & 1.00 \\
$\tw$ & 0.91 & 0.90 & 1.00 & 1.00 & 1.00 & 1.00 \\
\bottomrule
\end{tabular}

\end{table}

\paragraph{Parameter Values.}
Table~\ref{table-md-q-easy} shows the values of these two variants of each parameter relative to the corresponding parameter of the original graph.
Overall we can see the biggest improvement for the maximum degree $\Delta$ which suggests that high-degree vertices are often adjacent to many vertices with similar neighborhoods. Other parameters with sizable improvements are the closure number $c$, the $\core$ and $\core[3]$, the 1- and 2-bounded-degree deletion numbers, and the feedback vertex set number~$\fvs$: Here the median values decrease by up to $40\%$ while the median values for most other parameters decrease by less than $10\%$. Hence, for these parameters it seems worthwhile to lift tractability results from the standard parameters to their quotient-graph-based counterparts.

\begin{table}[t]
\tiny
\caption{Average, median, and 90th percentile value of the percentage increase from the exact values to the heuristic results on the set of $59$ instances where all of these values could be calculated exactly.}
\label{tab:stats_heuristics_comparison}
\centering
\begin{tabular}{l r r r}
\toprule
$k$ & avg. $k$ & median $k$ & 90th p. $k$ \\
\midrule
$\vc$ & 0.9 & 0.0 & 3.1 \\
$\bdd$ & 1.8 & 0.0 & 7.2 \\
$\bdd[2]$ & 4.3 & 2.2 & 13.1 \\
$\cvd$ & 3.4 & 0.0 & 12.7 \\
$\pvc$ & 11.3 & 1.5 & 29.5 \\
$\dco$ & 5.1 & 4.5 & 14.4 \\
\midrule
$\fvs$ & 1.3 & 0.0 & 4.8 \\
$\vi$ & 5.7 & 0.0 & 20.2 \\
$\tdp$ & 0.0 & 0.0 & 0.0 \\
$\tw$ & 0.0 & 0.0 & 0.0 \\
\bottomrule
\end{tabular}

\end{table}

\section{Heuristic Parameters}
\label{sec:eval-heuristic}
To obtain an overview for a larger set of instances, we resorted to heuristics when the exact algorithms failed to compute the parameter. Using this, we obtained results on the 85 remaining graphs which, as expected, are mostly among the larger instances in our benchmark set. The parameters for which we used heuristics are $\vc$, $\bdd$, $\bdd[2]$, $\cvd$, $\pvc$, $\dco$, $\fvs$, $\vi$, $\tdp$, and $\tw$. For $\tw$ we use the winning solver from the 2017~PACE challenge~\cite{DBLP:journals/jco/Tamaki19}, for $\tdp$ we use the second-place solver from the 2020~PACE challenge~\cite{Strass20}. These are anytime solvers that continue trying to find better solutions until they receive a stop signal. We set a timeout for~$30$ minutes on each instance. For the remaining parameter computations we use simple greedy algorithms. To see how good the heuristics are, we calculated the percentage increase from the exact values to the heuristic results on the set of $59$ graphs where we were able to calculate all parameters exactly. These results are listed in Table~\ref{tab:stats_heuristics_comparison}; overall, the quality is good enough to allow for consideration of these heuristic parameter values.

\subsection{Absolute Parameter Values} The distribution of the parameter values is illustrated in~\Cref{fig:violin-plots-heuristics}. 
\begin{figure*}[t!]
  \subfloat[Degree-related parameters]{%
    \includegraphics[width=.49\linewidth]{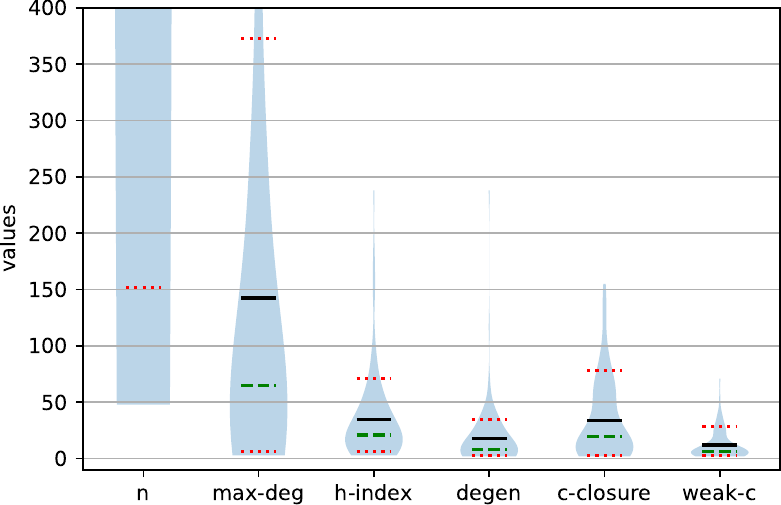}%
    \label{subfig:a}%
  }\hfill
  \subfloat[Neighborhood-based parameters]{%
    \includegraphics[width=.49\linewidth]{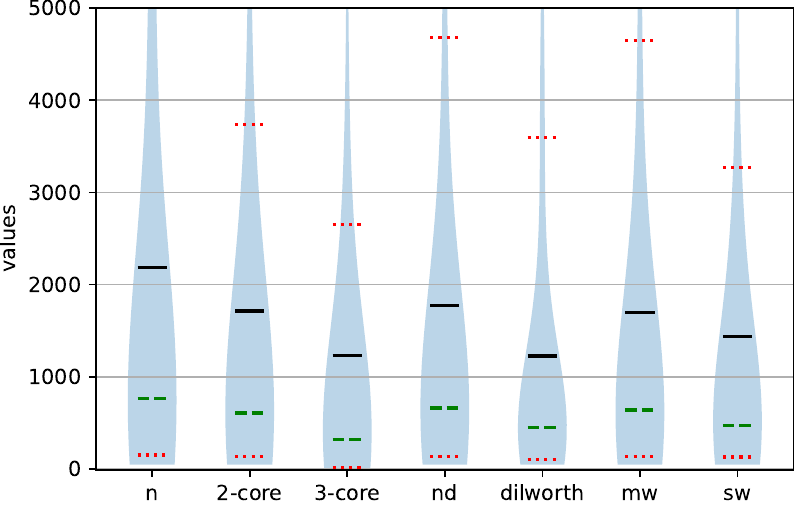}%
    \label{subfig:b}%
  }\\[4ex]
  
\centering  \subfloat[Modulator- and tree-based parameters]{%
    \includegraphics[width=.49\linewidth]{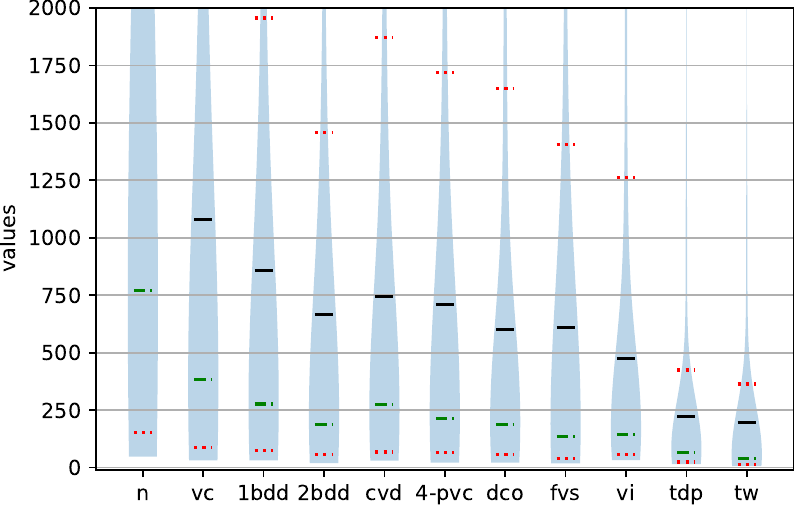}%
    \label{subfig:c}%
  }\hfill
  \caption{Distribution of the absolute parameter values on the 85~hard graphs.
    The thick lines indicate the 10th percentile (dotted red), the average (solid black), the median (dashed green), and the 90th percentile (dotted red), respectively.
    The shape of the violin visualizes the value distribution. 
  }
  \label{fig:violin-plots-heuristics}
\end{figure*}
In the degree-related parameter group, the overall picture is roughly the same as for the easy instance set. All parameters except the maximum degree take on reasonably small values for a large majority of the instances. The values of the maximum degree are also small for roughly half of the instances. The relative differences between~$n$ and the parameter values are now larger, which indicates that this parameter group grows sublinearly with~$n$.  
In the neighborhood-based parameter group, the values are quite large, and the relative differences between~$n$ and the parameter values are smaller than for the easy instances.
Intuitively, this means that the structure captured by these parameters is less pronounced in larger graphs. The two smallest parameters are Dilworth number and splitwidth, with the Dilworth number now being slightly smaller than splitwidth; both parameter values are usually too high to be useful, however.
Finally, in the modulator-based parameter group, the values are larger than in the degree-based group and smaller than in the neighborhood-based group. The values of the parameters~$\bdd[2]$,~$\cvd$,~$\pvc$,~$\dco$, and~$\fvs$ are roughly similar. The treewidth and treedepth are considerably smaller and their distribution is now clearly distinguishable from the other parameters. 
Moreover, the median values of both treewidth and treedepth are below~50 which shows that they are also practically very useful.

\subsection{Klam-Values}
\label{sec:klam}
 \begin{table}[t]
\tiny
\caption{Number of benchmark instances with~$k$ being below the Klam-Values for different running time bounds~$f(k)$. 
The total number of instances is 85.}
\label{table-para-statistic-klam-heuristic}
\centering
\begin{tabular}{l r r r r r r r r}
\toprule
parameter & $2^{k/\log(k)}$ & $\sqrt{2}^k$ & $2^k$ & $4^k$ & $k!$ & $k^k$ & $2^{k^2}$ & $2^{(2^k)}$ \\
\midrule
$\Delta$ & 83 & 57 & 43 & 32 & 28 & 21 & 12 & 9 \\
$h$-index & 85 & 81 & 76 & 58 & 43 & 38 & 15 & 10 \\
$d$ & 85 & 84 & 82 & 73 & 61 & 58 & 46 & 41 \\
$c$ & 85 & 83 & 71 & 54 & 46 & 40 & 27 & 23 \\
$\gamma$ & 85 & 85 & 84 & 78 & 67 & 64 & 49 & 43 \\
\midrule
$\core$ & 43 & 8 & 2 & 0 & 0 & 0 & 0 & 0 \\
$\core[3]$ & 55 & 20 & 12 & 11 & 11 & 9 & 9 & 9 \\
\midrule
$\nd$ & 41 & 9 & 2 & 0 & 0 & 0 & 0 & 0 \\
$\nabla$ & 45 & 14 & 4 & 0 & 0 & 0 & 0 & 0 \\
$\mw$ & 42 & 9 & 2 & 0 & 0 & 0 & 0 & 0 \\
$\sw$ & 46 & 10 & 2 & 0 & 0 & 0 & 0 & 0 \\
\midrule
$\vc$ & 59 & 19 & 6 & 2 & 0 & 0 & 0 & 0 \\
$\bdd$ & 63 & 27 & 6 & 3 & 0 & 0 & 0 & 0 \\
$\bdd[2]$ & 66 & 33 & 11 & 4 & 1 & 0 & 0 & 0 \\
$\cvd$ & 64 & 30 & 7 & 3 & 0 & 0 & 0 & 0 \\
$\pvc$ & 66 & 30 & 11 & 3 & 0 & 0 & 0 & 0 \\
$\dco$ & 66 & 34 & 13 & 4 & 0 & 0 & 0 & 0 \\
\midrule
$\fvs$ & 67 & 41 & 18 & 7 & 1 & 0 & 0 & 0 \\
$\vi$ & 70 & 40 & 12 & 2 & 0 & 0 & 0 & 0 \\
$\tdp$ & 77 & 59 & 44 & 15 & 5 & 1 & 0 & 0 \\
$\tw$ & 78 & 65 & 53 & 34 & 24 & 16 & 3 & 0 \\
\bottomrule
\end{tabular}

\end{table}
To get a clearer view of the relation between growth of~$f$ and the number of feasible instances for each parameter, we again make use of Klam-Values. Table~\ref{table-para-statistic-klam-heuristic} shows the results. As suggested by the parameter value distributions, degree-based parameterizations and treewidth and treedepth are the most promising. For treedepth and treewidth, already slightly superexponential running time bounds such as~$k^k$ are too high for a large majority of the instances. To guarantee acceptable running times for modulator-based parameters such as the vertex cover number, one essentially can only afford single-exponential running times~$\alpha^k$ with a small base~$\alpha$.

\subsection{Parameterizations Based on Quotient Graphs}
\label{sec:quotient}
\begin{table}[t]
\tiny
\caption{Values of the parameterization based on twin-quotient graphs,~$k_\nd$, and modular decompositions,~$k_\mw$, relative to the original parameter~$k$.}
\label{table-md-q-hard}
\centering

\begin{tabular}{l r r r r r r}
\toprule
 & \multicolumn{2}{c}{avg.} & \multicolumn{2}{c}{median} & \multicolumn{2}{c}{90th p.} \\[1ex]
 & \multicolumn{1}{c}{$k_{\nd} / k$}& \multicolumn{1}{c}{$k_{\mw} / k$} & \multicolumn{1}{c}{$k_{\nd} / k$} & \multicolumn{1}{c}{$k_{\mw} / k$} & \multicolumn{1}{c}{$k_{\nd} / k$} & \multicolumn{1}{c}{$k_{\mw} / k$} \\
\midrule
$n$ & 0.88 & 0.86 & 0.96 & 0.95 & 1.00 & 1.00 \\
\midrule
$\Delta$ & 0.91 & 0.91 & 1.00 & 1.00 & 1.00 & 1.00 \\
$h$-index & 0.96 & 0.96 & 1.00 & 1.00 & 1.00 & 1.00 \\
$d$ & 0.95 & 0.95 & 1.00 & 1.00 & 1.00 & 1.00 \\
$c$ & 0.95 & 0.95 & 1.00 & 1.00 & 1.00 & 1.00 \\
$\gamma$ & 0.99 & 0.98 & 1.00 & 1.00 & 1.00 & 1.00 \\
\midrule
$\core$ & 0.93 & 0.92 & 0.99 & 0.98 & 1.00 & 1.00 \\
$\core[3]$ & 0.84 & 0.83 & 0.99 & 0.98 & 1.00 & 1.00 \\
\midrule
$\nd$ & 0.98 & 0.97 & 1.00 & 1.00 & 1.00 & 1.00 \\
$\nabla$ & 0.98 & 0.96 & 1.00 & 1.00 & 1.00 & 1.00 \\
$\mw$ & 1.00 & 1.00 & 1.00 & 1.00 & 1.00 & 1.00 \\
$\sw$ & 1.00 & 1.00 & 1.00 & 1.00 & 1.00 & 1.00 \\
\midrule
$\vc$ & 0.94 & 0.92 & 0.99 & 0.99 & 1.00 & 1.00 \\
$\bdd$ & 0.94 & 0.93 & 0.97 & 0.99 & 1.00 & 1.00 \\
$\bdd[2]$ & 0.94 & 0.93 & 0.99 & 0.99 & 1.00 & 1.00 \\
$\cvd$ & 0.97 & 0.95 & 1.00 & 0.99 & 1.00 & 1.00 \\
$\pvc$ & 0.96 & 0.95 & 1.00 & 1.00 & 1.01 & 1.00 \\
$\dco$ & 1.00 & 0.98 & 1.00 & 1.00 & 1.01 & 1.00 \\
\midrule
$\fvs$ & 0.94 & 0.92 & 0.99 & 1.00 & 1.02 & 1.00 \\
$\vi$ & 0.96 & 0.97 & 0.98 & 0.99 & 1.01 & 1.00 \\
$\tdp$ & 0.99 & 0.99 & 1.00 & 1.00 & 1.00 & 1.00 \\
$\tw$ & 0.99 & 0.99 & 1.00 & 1.00 & 1.03 & 1.01 \\
\bottomrule
\end{tabular}

\end{table}

Finally, we consider the potential for improvements that can be obtained by using the parameters that are defined on twin-quotient graphs or on the quotient graphs of the modular decomposition. The results are shown in~\Cref{table-md-q-hard}. As could be expected from the large values of neighborhood diversity~$\nd$ and modular-width~$\mw$ on the hard instances, the improvement of~$k_{\nd}$ and~$k_{\mw}$ over~$k$ is usually only very small and virtually nonexistent for half of the instances. For some parameter combinations, it seems even that~$k_{\nd}>k$ but this is only because the parameter values are computed heuristically on some instances.  
\section{Conclusion}
\label{sec:conclusion}

We evaluated 21 popular graph parameters on 144 real-world graphs.
On one hand, some parameters such as neighborhood diversity~\nd or modular-width~\mw are surprisingly large for most graphs.
Thus, in order to improve upon~$f(n)$-time algorithms with the help of these parameters, one essentially needs to design algorithms with running time~$f(\nd)$ or~$f(\mw)$.
Of course studying these parameters is still well motivated if, for example, the resulting algorithms use interesting techniques.

On the other hand, treewidth~\tw, treedepth~\tdp and degree-based parameters that are smaller than the maximum degree~$\Delta$ such as $h$-index or degeneracy~$d$ are quite small on most of the graphs.
Hence, from a data-driven perspective the parameterized algorithmics community should focus on developing efficient FPT-algorithms for these parameters. The modulator-based parameters sit somewhere in between the large and small parameters. 
The best ones of these parameters, such as distance to cograph and feedback vertex set, take on promising values on some instances but these are much fewer than for \tw, \tdp, and the degree-based category. 

In addition, we considered combinations of treewidth with other parameters and showed that the combination of~$h$-index and treewidth is, from the data perspective, a good parameter combination that could be considered for problems that are hard when parameterized by treewidth alone. We also considered approaches of combining modular-width and neighborhood diversity with other parameters~$k$ by computing~$k$ on the modular decompositions and twin-quotient graphs, respectively. For the easier graphs, these two parameters gave some modest improvements over~$k$. In contrast, for the harder graphs the parameter values essentially stayed the same. 

In summary, our study shows that for our group of hard graphs (which are still modestly large compared to the massive graphs arising in many applications), the considered parameterizations are not fully satisfactory, and that there is still a need for other approaches to capture real-world structure. In other words, adding further parameters is a clear goal for future parameter reports. In particular, it would be interesting to extend the neighborhood-based category by smaller parameters such as twinwidth which had been excluded since current solvers for these parameters were not capable of solving a sufficient number of instances.


Another avenue for future work is to extend the data set on which we perform the study. Considering even larger graphs is possible but in our opinion less interesting at the moment, as we already have some modestly large graphs that have quite large values for all considered parameters. Instead, the goal will be to add further graphs of the considered sizes so that the number and diversity of graphs have increased sufficiently to distinguish the application background. For example, we may ask whether the treewidth of social networks is smaller or larger than the treewidth of infrastructure networks. With our current data set the single categories are, however, too small to provide a trustworthy answer to such questions.

One may also consider more sophisticated descriptive statistics for comparing parameters. For example, the currently considered median and percentile values of the relation~$k/n$ ignore whether an instance with~$k/n=0.1$ was large or small. In a similar direction, one could think of more sophisticated  models for estimating the tractable parameter value range other than the Klam-Values. 

Finally, it would be interesting to perform similar studies for other classes of inputs such as directed graphs, strings, or Boolean formulas. 

\section*{Acknowledgements}

We thank Michael Rohleder (Friedrich Schiller University Jena) for helping us with the implementation of the greedy algorithms for the heuristic parameter computations.
We also thank the reviewers and the editor from the Journal of Graph Algorithms and Applications (JGAA) for their valuable and detailed feedback and suggestions.

\bibliographystyle{plain}
\bibliography{refs}

\appendix
\newpage
\section{Instance Properties}
\label{sec:instances}
\begin{table}[h!]
\tiny
\caption{Sizes and some additional parameter values for all 59 \emph{easy} graphs for which all parameters could be calculated. A complete overview is available in the online repository.} 
\label{table-instance-spec-easy}
\centering

\begin{tabular}{l r r r r r}
\toprule
name & $n$ & $m$ & $\Delta$ & $\vc$ & $\tw$  \\
\midrule
nd\_7\_10 & 7 & 10 & 6 & 3 & 2 \\
florentine\_families & 15 & 20 & 6 & 8 & 3 \\
zebra & 27 & 111 & 14 & 20 & 12 \\
davis\_southern\_women & 32 & 89 & 14 & 14 & 8 \\
montreal & 33 & 78 & 18 & 12 & 6 \\
karate-club & 34 & 78 & 17 & 14 & 5 \\
insecta-ant-colony4-day35 & 35 & 501 & 33 & 31 & 28 \\
ceo\_club & 40 & 95 & 21 & 15 & 8 \\
elite & 44 & 99 & 12 & 20 & 9 \\
macaque\_neural & 47 & 313 & 27 & 33 & 16 \\
USA\_Mixed\_2016 & 48 & 103 & 8 & 30 & 6 \\
bcspwr02 & 49 & 59 & 6 & 22 & 3 \\
contiguous-usa & 49 & 107 & 8 & 30 & 6 \\
eco-stmarks & 54 & 350 & 48 & 36 & 22 \\
ENZYMES\_g103 & 59 & 115 & 9 & 33 & 5 \\
dolphins & 62 & 159 & 12 & 34 & 10 \\
terrorists\_911 & 62 & 152 & 22 & 29 & 7 \\
eco-everglades & 69 & 880 & 64 & 42 & 31 \\
mammalia-voles-bhp-trapping-49 & 78 & 119 & 9 & 44 & 6 \\
world-trade & 80 & 875 & 77 & 50 & 22 \\
baseball & 84 & 84 & 44 & 12 & 2 \\
tols90 & 90 & 1\,449 & 89 & 18 & 18 \\
bn-macaque-rhesus\_brain\_2 & 91 & 582 & 87 & 11 & 11 \\
bn-macaque-rhesus\_cerebral-cortex & 91 & 1\,401 & 87 & 29 & 28 \\
GD99\_c & 105 & 120 & 5 & 47 & 3 \\
centrality-literature & 118 & 613 & 66 & 59 & 24 \\
PACE\_2021\_jena\_3 & 121 & 1\,191 & 50 & 96 & 36 \\
PACE\_2021\_jena\_4 & 123 & 433 & 18 & 91 & 11 \\
PACE\_2021\_jena\_1 & 125 & 1\,319 & 58 & 96 & 31 \\
PACE\_2021\_jena\_2 & 134 & 2\,877 & 85 & 112 & 58 \\
revolution & 141 & 160 & 59 & 5 & 3 \\
PACE\_2021\_jena\_6 & 173 & 1\,039 & 35 & 126 & 29 \\
world\_Mixed\_2016 & 173 & 328 & 14 & 91 & 6 \\
interactome\_pdz & 212 & 242 & 21 & 56 & 6 \\
PACE\_2021\_strings\_1 & 236 & 1\,045 & 23 & 155 & 14 \\
hyves\_n & 278 & 302 & 28 & 77 & 4 \\
PACE\_2021\_colors\_2 & 293 & 573 & 24 & 168 & 13 \\
oryza\_sativa\_japonica & 327 & 341 & 234 & 25 & 3 \\
marvel\_partnerships & 350 & 346 & 12 & 144 & 4 \\
oryctolagus\_cuniculus & 350 & 333 & 81 & 66 & 4 \\
as\_skitter\_n & 375 & 439 & 155 & 10 & 4 \\
gallus\_gallus & 457 & 480 & 111 & 101 & 5 \\
human\_papillomavirus\_16 & 484 & 511 & 185 & 6 & 2 \\
hepatitus\_C\_Virus & 493 & 491 & 490 & 2 & 1 \\
human\_herpesvirus\_4 & 505 & 543 & 204 & 20 & 4 \\
nd\_511\_1830 & 511 & 1\,830 & 90 & 255 & 4 \\
diseasome & 516 & 1\,188 & 50 & 285 & 10 \\
canis\_familiaris & 529 & 538 & 406 & 30 & 3 \\
danio\_rerio & 581 & 602 & 244 & 87 & 3 \\
bos\_taurus & 605 & 585 & 81 & 164 & 4 \\
human\_Herpesvirus\_8 & 904 & 914 & 140 & 65 & 4 \\
middle-east\_coronavirus & 987 & 1\,051 & 540 & 24 & 4 \\
PhD & 1\,025 & 1\,043 & 46 & 259 & 3 \\
candida\_albicans\_SC5314 & 1\,191 & 1\,708 & 427 & 159 & 11 \\
board\_directors & 1\,217 & 1\,130 & 26 & 204 & 3 \\
moreno\_crime & 1\,380 & 1\,476 & 25 & 451 & 6 \\
human\_Immunodeficiency\_Virus\_1 & 1\,446 & 1\,732 & 477 & 10 & 7 \\
netscience & 1\,461 & 2\,742 & 34 & 899 & 19 \\
severe\_acute\_coronavirus & 1\,546 & 1\,906 & 259 & 30 & 16 \\
\bottomrule
\end{tabular}

\end{table}

\begin{table}[h!]
\tiny
\caption{Sizes and some additional parameter values for all 85 \emph{hard} graphs for which at least one parameter could not be calculated. A complete overview is available in the online repository.} 
\label{table-instance-spec-hard}
\centering

\begin{tabular}{l r r r r r}
\toprule
name & $n$ & $m$ & $\Delta$ & $\vc$ & $\tw$  \\
\midrule
flying & 48 & 284 & 22 & 36 & 20 \\
ar\_m8x8 & 64 & 112 & 4 & 32 & 8 \\
moreno\_highschool & 70 & 274 & 19 & 48 & 13 \\
polbooks & 105 & 441 & 25 & 62 & 13 \\
adjnoun\_adjacency & 112 & 425 & 49 & 59 & 28 \\
infect-hyper & 113 & 2\,196 & 98 & 90 & 73 \\
football & 115 & 613 & 12 & 94 & 34 \\
flower\_4\_1 & 129 & 372 & 8 & 80 & 36 \\
email-enron-only & 143 & 623 & 42 & 86 & 21 \\
digg\_n & 166 & 1\,375 & 70 & 86 & 41 \\
rel5 & 172 & 646 & 44 & 33 & 15 \\
PACE\_2021\_jena\_5 & 184 & 1\,367 & 53 & 136 & 23 \\
arenas-jazz & 198 & 2\,742 & 100 & 158 & 51 \\
scotland & 228 & 358 & 13 & 89 & 14 \\
can\_229 & 229 & 774 & 12 & 162 & 21 \\
saylr1 & 238 & 445 & 4 & 119 & 14 \\
dwt\_245 & 245 & 608 & 12 & 147 & 11 \\
PACE\_2021\_strings\_2 & 248 & 2\,451 & 38 & 209 & 23 \\
ar\_m16x16 & 256 & 480 & 4 & 128 & 16 \\
econ-wm1 & 258 & 2\,389 & 106 & 115 & 57 \\
c260 & 260 & 390 & 3 & 140 & 20 \\
as\_internet\_topology\_n & 266 & 6\,300 & 223 & 110 & 66 \\
PACE\_2021\_colors\_1 & 330 & 2\,256 & 43 & 260 & 33 \\
usair97 & 332 & 2\,126 & 139 & 149 & 33 \\
maier\_facebook\_friends & 348 & 1\,988 & 63 & 218 & 26 \\
poisson2D & 367 & 1\,025 & 8 & 248 & 20 \\
blogcatalog\_n & 372 & 2\,627 & 188 & 152 & 71 \\
BNU1\_0025884\_1\_DTI\_DS00446 & 425 & 4\,817 & 95 & 314 & 67 \\
celegans\_metabolic & 453 & 2\,025 & 237 & 249 & 32 \\
n3c6-b2 & 455 & 1\,313 & 16 & 91 & 67 \\
youtube\_n & 479 & 2\,428 & 80 & 245 & 107 \\
power-494-bus & 494 & 586 & 9 & 216 & 7 \\
us\_patents\_n & 502 & 1\,189 & 24 & 234 & 49 \\
NL\_Mixed\_2016 & 514 & 1\,618 & 26 & 341 & 17 \\
wikiconflicts\_n & 586 & 5\,704 & 297 & 228 & 66 \\
dwt\_607 & 607 & 2\,262 & 13 & 405 & 32 \\
social\_location\_gowalla\_n & 654 & 2\,104 & 212 & 190 & 52 \\
nos6 & 675 & 1\,290 & 4 & 337 & 15 \\
ppi\_interolog\_human\_n & 677 & 1\,611 & 231 & 131 & 30 \\
wordnet\_n & 692 & 1\,711 & 517 & 383 & 27 \\
can\_715 & 715 & 2\,975 & 104 & 507 & 41 \\
internet\_top\_pop\_kdl & 754 & 895 & 7 & 369 & 7 \\
socfb-Caltech36 & 769 & 16\,656 & 248 & 536 & 325 \\
BNU1\_0025913\_1\_DTI\_DS00833 & 780 & 9\,336 & 113 & 587 & 116 \\
reptilia-tortoise-network-fi & 787 & 1\,197 & 17 & 396 & 13 \\
copenhagen\_fb\_friends & 800 & 6\,418 & 101 & 531 & 249 \\
cpan\_perl\_module\_users & 839 & 2\,112 & 327 & 116 & 34 \\
kegg\_metabolic\_aae & 926 & 2\,334 & 188 & 441 & 34 \\
flixster\_n & 941 & 7\,745 & 150 & 600 & 158 \\
flickr\_n & 955 & 9\,448 & 207 & 510 & 228 \\
socfb-Reed98 & 962 & 18\,812 & 313 & 687 & 380 \\
bn-mouse-kasthuri\_graph\_v4 & 1\,029 & 1\,559 & 123 & 171 & 35 \\
arenas-email & 1\,133 & 5\,451 & 71 & 594 & 197 \\
euroroad & 1\,174 & 1\,417 & 10 & 571 & 13 \\
plasmodium\_falciparum\_3D7 & 1\,229 & 2\,448 & 51 & 450 & 119 \\
escherichia\_coli\_K12\_MG1655 & 1\,258 & 1\,875 & 58 & 452 & 48 \\
jagmesh9 & 1\,349 & 3\,876 & 6 & 899 & 21 \\
socfb-Simmons81 & 1\,518 & 32\,988 & 300 & 1\,106 & 576 \\
xenopus\_leavis & 1\,596 & 2\,054 & 640 & 212 & 12 \\
collins\_yeast & 1\,622 & 9\,070 & 127 & 1\,004 & 67 \\
bcspwr09 & 1\,723 & 2\,394 & 14 & 775 & 10 \\
biological\_pathways\_kegg & 1\,729 & 1\,924 & 18 & 782 & 12 \\
ppi-yeast & 1\,846 & 2\,203 & 56 & 626 & 43 \\
plant\_pol\_robertson & 1\,884 & 15\,255 & 300 & 411 & 273 \\
blckhole & 2\,121 & 6\,370 & 19 & 1\,423 & 75 \\
nd\_2453\_47659 & 2\,453 & 47\,659 & 444 & 1\,429 & 117 \\
socfb-Trinity100 & 2\,613 & 111\,996 & 404 & 2\,155 & 1\,456 \\
dwt\_2680 & 2\,680 & 11\,173 & 18 & 1\,977 & 43 \\
sstmodel & 2\,730 & 9\,702 & 17 & 1\,911 & 37 \\
route\_views & 3\,015 & 5\,156 & 590 & 566 & 39 \\
lshp3466 & 3\,466 & 10\,215 & 6 & 2\,310 & 72 \\
california & 4\,271 & 8\,909 & 175 & 892 & 235 \\
nyc\_restaurant\_checkin & 4\,936 & 13\,472 & 88 & 2\,010 & 853 \\
opsahl-powergrid & 4\,941 & 6\,594 & 19 & 2\,203 & 17 \\
grqc & 5\,241 & 14\,484 & 81 & 2\,783 & 174 \\
c6000 & 6\,000 & 9\,000 & 3 & 3\,176 & 135 \\
rattus\_norvegicus & 6\,147 & 10\,387 & 1\,305 & 847 & 100 \\
nyc\_restaurant\_tips & 6\,410 & 9\,472 & 196 & 2\,491 & 340 \\
erdos992 & 6\,927 & 11\,850 & 507 & 474 & 148 \\
ppi-dmela & 7\,393 & 25\,569 & 190 & 2\,630 & 1\,037 \\
us\_roads\_DC & 9\,559 & 14\,841 & 6 & 5\,068 & 34 \\
hepth & 9\,875 & 25\,973 & 65 & 4\,981 & 725 \\
hepph & 12\,006 & 118\,489 & 491 & 7\,012 & 1\,230 \\
astroph & 18\,771 & 198\,050 & 504 & 12\,012 & 3\,329 \\
condmat & 23\,133 & 93\,439 & 279 & 13\,521 & 2\,004 \\
\bottomrule
\end{tabular}

\end{table}

\section{Algorithms}
\label{sec:algos}

In this appendix we briefly explain what algorithms we use to compute the various parameters and how they are implemented. 
The algorithms for most of the simple parameters that can be computed in polynomial time are directly implemented in Python with the help of the NetworkX package~\cite{networkx}.
These algorithms are described in~\Cref{sec:simple_params_algos}. 
For some of the parameters that are NP-hard to compute, we used existing solvers written in Java, C, or C++. 
We will list these solvers in~\Cref{sec:existing_solvers_params_algos}. Finally, to compute the remaining parameters, we use ILP-formulations that are then solved by the ILP-solver Gurobi~\cite{gurobi}. 
These will be explained in~\Cref{sec:ilp_params_algos}.

\subsection{Polynomial-time Computable Parameters}
\label{sec:simple_params_algos}

To compute the parameters $n$, $m$, the maximum degree~$\Delta$, the degeneracy~$d$, and the $k$-core size we use existing implementations that are part of the NetworkX package~\cite{networkx}. 
For the parameters $h$-index, closure number, and neighborhood diversity we implemented simple and straightforward algorithms that directly follow from the definition of these parameters. 
To compute the Dilworth number, we use the known reduction of the problem to a bipartite maximum matching instance~\cite{fulkerson1956,FRS03}. This instance is solved by a method provided by NetworkX.

To compute the weak closure number, we first fix some number~$\gamma$. 
We then continuously remove any vertex~$v$ from~$G$ that does not share at least~$\gamma$ neighbors with any nonneighbor of~$v$. 
A graph~$G$ is weakly~$\gamma$-closed if and only if the graph is empty after this process~\cite{FRSWW20}. 
Hence, we can compute the weak closure number of~$G$ by performing a binary search for~$\gamma$. 
We use the minimum of $d+1$ and the closure number~$c$ of~$G$ as an initial upper bound for~$\gamma$.

\subsubsection{Splitwidth}

Our implementation for computing the splitwidth is based on an algorithm developed by Cunningham~\cite{Cun82} for directed graphs that runs in $\Oh(n^3)$~time and is trivially adaptable to undirected graphs.
Even though there are linear-time algorithms to compute the splitwidth of an undirected graph~\cite{CMR12}, we decided to implement the comparably slower algorithm by Cunningham~\cite{Cun82} based on the simplicity of the algorithm.

Intuitively, Cunningham's algorithm~\cite{Cun82} works as follows:
For a given graph~$G$, the algorithm tries to find a simple decomposition of~$G$, if one exists.
If no such simple decomposition exists, that is, if~$G$ is prime, the algorithm terminates.
Otherwise, if there is a simple decomposition~$(G_1,G_2)$ of~$G$, the algorithm performs the above procedure for both~$G_1$ and~$G_2$ recursively. 
To determine, whether there is a simple decomposition~$(G_1,G_2)$ of~$G$, the algorithm iterates over all edges~$e$ of~$G$ and checks whether there is a split~$(V_1,V_2)$ of~$G$ that is ``crossed'' by~$e$.
Here, we say that an edge~$e$~\emph{crosses} a split~$(V_1,V_2)$ of~$G$ if each of the parts of the split contains exactly one endpoint of~$e$.
Checking whether there is a split~$(V_1,V_2)$ of~$G$ that is crossed by~$e$ (and if so, finding one) can be done efficiently~\cite{Cun82}.

To speed up our implementation, we prevent the algorithm from checking whether an edge~$e$ crosses a split in any subgraph, if we already checked that~$e$ does not cross any split of the supergraph.
This is based on the fact that, if an edge~$e$ of~$G$ crosses no split of~$G$, then for each simple decomposition~$(G_1,G_2)$ of~$G$, $e$ does not cross any split of~$G_1$ and $e$ does not cross any split of~$G_2$. Below, we give a proof of this fact.
\begin{lemma}
Let~$G$ be a graph, let~$(G_1,G_2)$ be a simple decomposition of~$G$, and let~$\{u,v\} \in E(G) \cap E(G_1)$.
If there is a split~$(W_1,W_2)$ of~$G_1$ with~$u\in W_1$ and~$v\in W_2$, then there is a split of~$G$ that is crossed by~$\{u,v\}$. 
\end{lemma}
\begin{proof}
Let~$x$ be the unique vertex of~$V(G_1) \setminus V(G)$.
Without loss of generality assume that~$x \in W_1$.
We define a partition~$(V_1',V_2')$ of~$V(G)$ with~$u\in V'_1$ and~$v\in V'_2$ as follows:
We set~$V_1' \coloneqq (W_1 \cup V(G_2)) \setminus \{x\}$ and~$V_2' \coloneqq V(G) \setminus V_1' = W_2$.
In the following, we show that~$(V_1',V_2')$ is a split of~$G$.
Since~$(W_1,W_2)$ is a split of~$G_1$, $|W_1| \geq 2$ and~$|W_2| \geq 2$.
Hence, $V_2'$ has size at least two.
Moreover, since~$(G_1, G_2)$ is a simple decomposition of~$G$, $V(G_2)$ contains at least three vertices.
Consequently, $V_1'$ has size at least two.
Hence, to show that~$(V_1',V_2')$ is a split of~$G$, it remains to show that each vertex of~$Y \coloneqq \{w\in V_1' : \exists w' \in V_2'.\, \{w, w'\} \in E(G)\}$ has the same neighborhood in~$V_2' = W_2$.
Let~$y$ and~$z$ be distinct vertices of~$Y$.
We show that~$N(y) \cap V_2' = N(z) \cap V_2'$.
To this end, we distinguish three cases about the location of~$y$ and~$z$.

\textbf{Case 1:} Both~$y$ and~$z$ are contained in~$V(G_1)$\textbf{.}
Then, the statement holds by the fact that~$(W_1,W_2)$ is a split of~$G_1$ with~$\{y,z\} \subseteq V_1'$ and~$V_2' = W_2$.

\textbf{Case 2:} Both~$y$ and~$z$ are contained in~$V(G_2)$\textbf{.}
Then, the statement holds by the fact that~$W_2$ contains no vertex of~$V(G_2)$ and each vertex of~$V(G_2)$ with at least one neighbor in~$V(G_1) \setminus\{x\}$ has the same neighborhood in~$V(G_1) \setminus\{x\}$.
The latter holds, since $(G_1, G_2)$ is a simple decomposition of~$G$.

\textbf{Case 3:} Without loss of generality~$y$ is contained in~$V(G_1)$ and~$z$ is contained in~$V(G_2)$\textbf{.}
Since $(G_1, G_2)$ is a simple decomposition of~$G$, $z$ has the same neighborhood in~$V(G_1)$ as~$x$.
Moreover, since~$z$ is contained in~$Y$, $z$ has at least one neighbor in~$W_2$.
Hence, $x$ and~$z$ have the same neighborhood in~$W_2 = V'_2$ as~$y$ since~$(W_1,W_2)$ is a split of~$G_1$.

Hence, each vertex of~$Y$ has the same neighborhood in~$W_2$.
This implies that, $(V_1', V_2')$ is a split of~$G$.
\end{proof}





\subsection{Existing Solvers}
\label{sec:existing_solvers_params_algos}

To compute the parameters modular-width, feedback vertex set number, treedepth, and treewidth we use existing solvers.
For modular-width we use a solver by Mizutani who implemented a linear-time algorithm by Tedder et al.~\cite{DBLP:conf/icalp/TedderCHP08}. The original code can be found at \url{https://github.com/mogproject/modular-decomposition}.
For feedback vertex set we use the winning solver from the 2016 PACE challenge~\cite{DHJ+16} by Iwata and Imanishi. Their solver first computes a linear-time kernel~\cite{DBLP:conf/icalp/Iwata17} and then uses an FPT branch-and-bound algorithm~\cite{IWY16}. The original code can be found at \url{https://github.com/wata-orz/fvs}.
For treedepth we use an updated version of the winning solver from the 2020 PACE challenge~\cite{KMN+20} by Trimble~\cite{Trim20}. The original code can be found at \url{https://github.com/jamestrimble/bute}.
Finally, for treewidth we use a solver by Tamaki~\cite{DBLP:conf/iwpec/Tamaki23}. The original code can be found at \url{https://github.com/twalgor/RTW}.

\subsection{ILP-Formulations for the Remaining Parameters}
\label{sec:ilp_params_algos}

We use standard ILP-formulations of the respective optimization problems to compute the vertex subset parameters from  \Cref{sec:subset_params} with the exception of feedback vertex set for which we use the above-mentioned solver by Iwata and Imanishi. To solve these formulations we then use  Gurobi~\cite{gurobi}. In this section we will explain these ILP-formulations.

Generally, the ILP-formulations for all these problems work similarly, since each of these problems involves finding a minimum-size modulator~$S$ to some graph property~$\Pi$:
For each vertex~$v\in V(G)$, we introduce a binary variable~$x_v$.
The variable~$x_v$ should be set to~$1$ by the ILP if and only if~$v$ is contained in the sought minimum-size modulator~$S$. 

We now just need to add constraints to ensure that~$G-S$ fulfills property~$\Pi$.
In the following, we describe how these constraints look for the individual parameters.

\begin{itemize}
\item For the vertex cover number, we add for each~$\{ u, v \}\in E(G)$ the constraint~$x_u + x_v \geq 1$.
\item For the~$r$-bounded-degree deletion number, we need to ensure that for each vertex~$v\in V(G)$ we include~$v$ in~$S$ or we include at least~$\deg_G(v) - r$ neighbors of~$v$ in~$S$. 
To this end, we add for each vertex~$v\in V(G)$ the constraint
\[ \deg_G(v) - r \leq \deg_G(v) \cdot x_v + \sum_{n \in N_G(v)} x_n. \]
\item For the cluster vertex deletion number, we add the constraint~$1 \leq x_u + x_v + x_w$ for each induced path~$(u, v, w)$ in~$G$.
\item For distance to cograph, we add the constraint~$1 \leq \sum_{v\in V(P)} x_v$ for every induced path~$P$ of length~$3$.
\item For~$4$-path vertex cover we add the constraint~$1 \leq \sum_{v\in V(P)} x_v$ for every (not necessarily induced) path~$P$ of length~$3$. 
\end{itemize}

Additionally, for cluster vertex deletion and $4$-path vertex cover we used some straightforward reduction rules and added some additional constraints to the ILP for small-degree vertices to slightly speed up the computation of these parameters.

Note that for~$4$-path vertex cover, we might add $\Theta(n^{4})$~constraints to the ILP.
To avoid this, we add most of the constraints in a lazy fashion. 
That is, initially we only add a small subset of the constraints.
Afterwards, we let Gurobi find an optimal solution~$S$ to the currently added constraints. 
If~$S$ is not a modulator to the respective graph property~$\Pi$, then at least one (currently not added) constraint is violated.
In this case, we add a small number of new constraints to the ILP that are currently violated.
This procedure is repeated until, eventually, a minimum-size modulator is found.



\subsubsection{Vertex Integrity}\label{sec:vi_algos}
We compute the vertex integrity via ILPs for a related parameter, the~\emph{$r$-component order connectivity number} of a graph~$G$. This parameter, denoted by~$\coc[r](G)$, is the size of a smallest modulator of~$G$ to the graph property that consists of all graphs $G'$ with $\cc(G') \leq r$ where $\cc(G')$ is the number of vertices in the largest connected component of $G'$~\cite{DDH16}.
Recall that~$\vi(G) \coloneqq \min\{|X|+ \cc(G-X) : X \subseteq V(G)\}$~\cite{BES87,DDH16,GHK+24}. An equivalent way to define the vertex integrity via $\coc[r](G)$ is $\vi(G)\coloneqq\min\{r + \coc[r](G):r\in \mathds{N}\}$.
Hence, to compute $\vi(G)$, we simply compute $\coc[r](G)$ for all values of $r\in [ub]$ in increasing order where $ub$ is an upper bound for $\vi(G)$. As initial upper bound we set $ub \coloneqq \cc(G)$~\cite{goddard1990integrity}.

To compute $\coc[r](G)$, we use the above-described general ILP approach and add the constraint~$1 \leq \sum_{v\in C} x_v$ for every connected subgraph~$C$ of~$G$ with exactly $r + 1$~vertices. Since the number of constraints becomes prohibitively large already for intermediate values of~$r$, we add the constraints in a lazy fashion as described above for the 4-path vertex cover number.

After computing $\coc[r](G)$ for some~$r$, we can update $ub$ by setting~$ub:=\min(ub,r+\coc[r](G))$. Once we have $r = ub$, we know that $ub$ must be equal to $\vi(G)$. We can also use $ub - 1 - r$ as an upper bound for the objective value of the ILPs. Since we only care about improving $ub$ until it is equal to $\vi(G)$ it is fine if this leads to some of the ILP formulations being infeasible. We also use $h(G) + 1 - r$ as a lower bound for the objective value of the ILPs~\cite{goddard1990integrity}.
We denote this algorithm as \texttt{vi}~\texttt{basic}.

Our experiments show that \texttt{vi}~\texttt{basic} is not fast enough to compete with the solvers for the other parameters. 
Because of this, we developed several improvements for our implementation in the form of an initial graph reduction or additional constraints for the ILP formulations. We will use this subsection to explain and show the correctness of these improvements.

The first two improvements let us identify some vertices that are either always part of a minimum-size modulator for $\coc[r](G)$ or never. 
We can use these to reduce the number of variables we add to each ILP instance.

\begin{lemma}[Folklore]
	Let $G$ be a graph and let~$r$ be some integer. Any vertex $v\in V(G)$ that is contained in a connected component $C$ with $|C| \leq r$ is never part of a minimum-size modulator for $\coc[r](G)$.
\end{lemma}

For Lemma~\ref{lem:vi_large_set_neigh}, recall that a vi-set of a graph~$G$ is any vertex set~$X$ with~$\vi(G) = |X| + \cc(G - X)$.

\begin{lemma}\label{lem:vi_large_set_neigh}
	Let $G$ be a graph and let~$C$ be a connected subset of $V(G)$ with $|N[C]| > \vi(G)$. Then each vi-set of $G$ contains at least one vertex of $N[C]$.
\end{lemma}
\begin{proof}
	Assume towards a contradiction that there is a vi-set~$X$ of~$G$ with $X\cap C = \emptyset$. 
	Note that since~$C$ is connected and since~$X\cap C=\emptyset$, all vertices in $C$ must be in the same connected component~$C'$ of~$G-X$.
	Now, any vertex in $N[C]$ is either in~$C'$ or in~$X$ which implies~$|X| + |C'| \geq |N[C]| > \vi(G)$. 
	But since~$X$ is a vi-set of~$G$, $\vi(G) = |X| + \cc(G -X) \geq |X| + |C'|$; a contradiction.
\end{proof}

The special case of~\Cref{lem:vi_large_set_neigh} for~$|C| = 1$ was originally shown by Drange et al.~\cite[Theorem 6]{DDH16}. 
We can use this special case to identify vertices that are part of any vi-set.
 Since we do not know the exact value of $\vi(G)$ while the algorithm is running, we instead use the current best upper bound. 
 When we would add constraints for a connected vertex set~$D$ of size at most~$r + 1$ to the ILP, we can keep track of the size of smaller connected subsets~$C\subsetneq D$. 
 If the size of the closed neighborhood of such a vertex set~$C$ exceeds the current upper bound for~$\vi(G)$, we add the constraint only for the vertex set~$C$ to the ILP.
 Due to \Cref{lem:vi_large_set_neigh}, this is correct.

Next, we use a result by Gima et al.~\cite[Lemma~2.2]{GHK+24}. 
For any vertex set~$S\subseteq V(G)$, a vertex~$v \in S$ is \emph{redundant} if at most one connected component of~$G - S$ contains neighbors of~$v$. 
Gima et al.~\cite{GHK+24} showed that there always exists a vi-set that contains no redundant vertices. 
The correctness of this statement follows from the observation that if~$v$ is a redundant vertex of~$S$, then removing~$v$ from~$S$ increases the size of~$\cc(G - S)$ by at most one, since all neighbors of~$v$ are in at most one connected component of~$G - S$. 
In fact, since we know that~$S$ is a vi-set, the size of~$\cc(G - S)$ must increase by exactly one.
\Cref{lem:vi_non_redundant_props} states some properties that directly follow from this result.

\begin{lemma}\label{lem:vi_non_redundant_props}
	Let $G$ be a graph and let~$S$ be a vi-set that contains no redundant vertex. 
	Then, 
	\begin{enumerate}
		\item\label{lem:vi_non_redundant_props_1} there is no vertex~$v\in V(G)$ with~$N[v] \subseteq S$,
		\item\label{lem:vi_non_redundant_props_2} 
	 the vertex set~$S$ contains no simplicial vertex, and
		\item\label{lem:vi_non_redundant_props_3} 
		for each pair of vertices~$v_1, v_2\in V(G)$ with~$N(v_1)\setminus S\subseteq N[v_2]$, we have $v_1\in S \Rightarrow v_2\in S$.
	\end{enumerate}
\end{lemma}
\begin{proof}
	We prove the statement by contradiction. 
	To this end, we show that~$S$ contains a redundant vertex, if at least one of the three properties is not fulfilled.
	
	\Cref{lem:vi_non_redundant_props_1}: 
	Let $v$ be a vertex with $N[v] \subseteq S$. In this case the neighborhood of $v$ in~$G - S$ is empty which implies that by definition, $v$ is a redundant vertex.
	
	\Cref{lem:vi_non_redundant_props_2}: 
	Let $v$ be a simplicial vertex in~$S$. 
	Since the neighborhood of $S$ in~$G$ is a clique, the neighbors of~$v$ that remain in~$G - S$ are still a clique. 
	Hence, $v$ is a redundant vertex.
	
	\Cref{lem:vi_non_redundant_props_3}: 
	Let $v_1, v_2 \in V(G)$ be a pair of vertices fulfilling (i)~$N(v_1)\setminus S\subseteq N[v_2]$, (ii)~$v_1\in S$, and (iii)~$v_2\notin S$. 
	Then each neighbor of $v_1$ in $G - S$ is also a neighbor of $v_2$ in $G - S$. 
	Thus, the entire neighborhood of $v_1$ in $G - S$ is part of the same connected component.
	Consequently, $v_1$ is a redundant vertex.
\end{proof}

It would be difficult to add constraints to the ILP enforcing the property that the solution set is not allowed to contain redundant vertices. Instead, we add constraints that enforce the three properties of \Cref{lem:vi_non_redundant_props}. 
For~\Cref{lem:vi_non_redundant_props_1}, for every vertex $v\in V(G)$ we add the constraint 
$$\deg(v) \geq \sum_{w\in N[v]} x_w.$$ 
We can further use~\Cref{lem:vi_non_redundant_props_2} to reduce the number of variables in the ILP by not including simplicial vertices. 
Finally, for~\Cref{lem:vi_non_redundant_props_3}, we do the following:
For each two distinct vertices~$v_1, v_2\in V(G)$, let $R \coloneqq N(v_1) \setminus N[v_2]$. 
\Cref{lem:vi_non_redundant_props_3} states that if the vi-set~$S$ contains at least all vertices of~$R \cup \{v_1\}$, then~$S$ should also contain vertex~$v_2$.
This can be expressed by the constraint
$$x_{v_1} \leq x_{v_2} + |R| - \sum_{v\in R} x_v.$$ 
%
We could theoretically add this constraint for every pair of vertices. 
However, we decided to only add the constraint for pairs for which~$R$ has size at most three.

Our final improvement is based on the following observation. Let~$r$ be some number and let~$v$ be a cut vertex in~$G$, that is, a vertex that if removed from~$G$ increases the number of connected components by at least one. 
If removing~$v$ from~$G$ creates a new connected component~$C$ with~$|C| \leq r$ then it is always better to include~$v$ in a modulator for~$\coc[r](G)$ instead of including any vertex in~$C$. 
To generalize this idea, we use the following definition.
To this end, recall that for a vertex set~$C$, $\conn(C)$ denotes the smallest integer~$k$ for which there is a vertex set~$F\subseteq C$ of size~$k$, such that~$G[C]-F$ is not connected.

\begin{definition}\label{def:cut_decomposition}
	For a graph~$G$ and an integer~$r$, a family~$\mathcal{C}$ of pairs of vertex sets is called an \emph{$r$-cut decomposition} of $G$, if each pair~$(C,\cut(C)) \in \mc$ fulfills the following properties: 
	\begin{enumerate}
	\item[a)] $C$ has size at most~$r$, 
	\item[b)] $\cut(C)\subseteq N(C)$, 
	\item[c)] $\conn(C\cup\cut(C))\ge|\cut(C)|$, and 
	\item[d)] for each two pairs~$(C_1,\cut(C_1)), (C_2,\cut(C_2))\in \mc$, we have~$C_1 \cap \cut(C_2) = \emptyset$ or~$C_2\subseteq C_1$.
	\end{enumerate}
\end{definition}

One possible $r$-cut decomposition consists of the pairs~$(\{v\}, D)$ for some fixed vertex~$v\in V(G)$ combined with all sets~$D=\cut(\{v\})$ which fulfill the requirements of  \Cref{def:cut_decomposition}. 
However, in most graphs this will only cover a small part of the graph. 
Ideally we would like to find an $r$-cut decomposition, where each vertex of the graph is part of some set $C$ or some set $\cut(C)$. 
For each pair~$(C, \cut(C))\in \mc$, we ideally also want that~$N(C) \setminus \cut(C)$ is as small as possible. 
The reason for this is the following lemma, which describes how an $r$-cut decomposition can be exploited in the ILP formulation.

\begin{lemma}\label{lem:vi_cut_sets}
Let $\mathcal{C}$ be an $r$-cut decomposition of~$G$.
There is a minimum-size modulator~$S$ for $\coc[r](G)$ such that for each pair $(C, \cut(C))\in \mathcal{C}$, the following conditions hold.
\begin{enumerate}
	\item\label{lem:vi_cut_sets_prop_1} $N(C) \setminus S \subseteq \cut(C) \Rightarrow C\cap S = \emptyset$	
	\item\label{lem:vi_cut_sets_prop_2} $(N(C) \setminus S \subseteq \cut(C) \wedge |C| = r) \Rightarrow \cut(C) \subseteq S$
\end{enumerate}
\end{lemma}
\begin{proof}
Let~$S$ be a minimum-size modulator for $\coc[r](G)$.

We first show that \Cref{lem:vi_cut_sets_prop_2} directly follows from \Cref{lem:vi_cut_sets_prop_1}. For this let~$(C,\cut(C))\in \mc$ be a pair. If $|C| < r$ or $N(C) \setminus S \not\subseteq \cut(C)$ then \Cref{lem:vi_cut_sets_prop_2} trivially holds. This means we can now assume that the left side of \Cref{lem:vi_cut_sets_prop_2} is true. This also means that the left side of \Cref{lem:vi_cut_sets_prop_1} must be true, which implies that the right side of \Cref{lem:vi_cut_sets_prop_1} is true. If~$C$ is a connected vertex set in~$G-S$, then we must have~$N(C)\subseteq S$ and in particular~$\cut(C) \subseteq S$ since we know that~$S$ is a modulator for~$\coc[r](G)$ and~$|C| = r$. If~$C$ is not connected in~$G-S$, then we know that~$C\cup \cut(C)$ is connected in~$G$ with~$\conn(C\cup \cut(C)) \geq |\cut(C)|$ (see Property c)). This implies that~$\cut(C) \subseteq S$ since at least that many vertices must be removed in order to disconnect~$C$ in~$G-S$.
Hence, \Cref{lem:vi_cut_sets_prop_2} follows from \Cref{lem:vi_cut_sets_prop_1} as long as $S$ is a modulator for $\coc[r](G)$.

If~\Cref{lem:vi_cut_sets_prop_1} 
 holds for each pair~$(C,\cut(C))\in \mc$ with respect to~$S$, the statement follows.
 Hence, assume that this is not the case and let~$(C,\cut(C))$ be a pair of~$\mc$ that does not fulfill~\Cref{lem:vi_cut_sets_prop_1} with respect to~$S$.



Let~$S'$ be a vertex set obtained from~$S$ by firstly removing all vertices of~$C$ and secondly adding~$\min(|\cut(C)|, |C\cap S|)$~vertices of $\cut(C)\setminus S$ if $N(C) \setminus S \subseteq \cut(C)$.
In the following, we show that~$S'$ is a minimum-size modulator for~$\coc[r](G)$ such that strictly more pairs of~$\mc$ fulfill~\Cref{lem:vi_cut_sets_prop_1}.
First, we show that~$S'$ is a minimum-size modulator for~$\coc[r](G)$.
By construction, $|S'|\leq |S|$.
Hence, we only have to show that~$S'$ is a modulator for~$\coc[r](G)$.

	Without loss of generality, we assume that $N(C) \setminus S \subseteq \cut(C)$ holds since this is the only case where~$S \not= S'$. 	
	If $|(C\cup \cut(C)) \cap S| \geq |\cut(C)|$ all vertices of $\cut(C)$ are contained in~$S'$ and all vertices of $C$ are removed from $S$, that is, $S'\cap C=\emptyset$. 
	Observe that~$|S'|\le |S|$ in this case and $\cc(G - S) \leq r$ still holds since $|C| \leq r$. 
	Also, notice if we have $|C| = r$ the condition~$|(C\cup \cut(C)) \cap S| \geq |\cut(C)|$ must always be true because $\conn(C \cup \cut(C)) \geq |\cut(C)|$. Hence, we have $\cut(C) \subseteq S$.
	
	Otherwise, if~$|(C\cup \cut(C)) \cap S| < |\cut(C)|$, we know that the remaining vertices in~$(C\cup \cut(C)) \setminus S$ are all part of the same connected component in $G - S$. Hence, adding any $|C \cap S|$~vertices from $\cut(C) \setminus S$ to $S$ and then removing all vertices in $C$ from $S$ preserves the property $\cc(G - S) \leq r$ and~$|S'|=|S|$.

By definition of~$S'$, $(C,\cut(C))$ trivially fulfills~\Cref{lem:vi_cut_sets_prop_1} with respect to~$S'$.
Next, we show that for each pair~$(D,\cut(D))\in \mc$, $(D,\cut(D))$ fulfills the right side of~\Cref{lem:vi_cut_sets_prop_1}
with respect to~$S'$ if $(D,\cut(D))$ fulfills the right side of~\Cref{lem:vi_cut_sets_prop_1} 
with respect to~$S$.
This then implies that the above-described procedure can be recursively applied and eventually constructs a minimum-size modulator~$S^*$ for $\coc[r](G)$ such that each pair of~$\mathcal{C}$ fulfills~\Cref{lem:vi_cut_sets_prop_1} (and thus also~\Cref{lem:vi_cut_sets_prop_1})  
with respect to~$S^*$, since with each application, at least one additional pair of~$\mc$ fulfills~\Cref{lem:vi_cut_sets_prop_1} with respect to the newly constructed modulator.

Let~$(D,\cut(D))\in \mc$, such that $(D,\cut(D))$ fulfills the right side of~\Cref{lem:vi_cut_sets_prop_1}
with respect to~$S$.
We show that~$(D,\cut(D))$ fulfills~\Cref{lem:vi_cut_sets_prop_1} 
with respect to~$S'$. 
Assume towards a contradiction that this is not the case.
We distinguish three cases.

\textbf{Case $1$:} $C\subseteq D$\textbf{.}
Recall that~$(C,\cut(C))$ does not fulfill the right side of \Cref{lem:vi_cut_sets_prop_1} with respect to~$S$.
Hence, $(D,\cut(D))$ does not fulfill the right side of \Cref{lem:vi_cut_sets_prop_1} with respect to~$S$, since~$C\subseteq D$; a contradiction to the choice of~$(D,\cut(D))$.

	\textbf{Case $2$:} $D\subsetneq C$\textbf{.} 
Since~$S'$ contains no vertex of~$C$, $S'$ contains no vertex of~$D$.
This implies that~$(D,\cut(D))$ fulfills the right side of \Cref{lem:vi_cut_sets_prop_2} with respect to~$S'$; a contradiction.
%
	
	\textbf{Case $3$:} $D\not\subseteq C$ and $C\not\subseteq D$\textbf{.} 
	In this case, $C \cap \cut(D) = \emptyset$ and $D \cap \cut(C) = \emptyset$ hold (see Property~d)). 
	Hence, we never add vertices in $D$ to $S$, which implies that~$(D,\cut(D))$ fulfills the right side of~\Cref{lem:vi_cut_sets_prop_2} with respect to~$S'$, since~$(D,\cut(D))$ fulfills the right side of~\Cref{lem:vi_cut_sets_prop_2} with respect to~$S$; a contradiction to the choice of~$(D,\cut(D))$.
\end{proof}

To compute an $r$-cut decomposition we do the following:
Let $\mathcal{D}$, $\mathcal{S}$, and $\mathcal{Q}$ be families of subsets of $V(G)$. 
Initially, $\mathcal{D}$ and $\mathcal{S}$ both contain the vertex set of every connected component in $G$, and~$\mathcal{Q}$ is empty. 
We now choose some subset $S\in \mathcal{S}$ such that $G[S]$ is not a clique and find a set of vertices $Q$ that is a minimum vertex cut of $G[S]$. We then remove $S$ from~$\mathcal{S}$, add the connected components of $G[S \setminus Q]$ to $\mathcal{D}$ and $\mathcal{S}$, and add~$Q$ to~$\mathcal{Q}$. 
We repeat this procedure until all subsets in $\mathcal{S}$ induce a clique in~$G$. 
Finally, for each two vertex sets~$D\in \mathcal{D}, Q\in \mathcal{Q}$, we add~$(D,Q)$ to~$\mathcal{C}$ if (i)~$|D| \leq r$, (ii)~$Q\subseteq N_G(D)$, and (iii)~$\conn(D \cup Q) \geq |Q|$. 
A composition created in that way is called an \emph{$r$-min-cut decomposition}.
Next, we verify that an \emph{$r$-min-cut decomposition} is indeed an $r$-cut decomposition.

\begin{lemma}
	An $r$-min-cut decomposition is an $r$-cut decomposition.
\end{lemma}
\begin{proof}
	The Properties a), b) and c) of~\Cref{def:cut_decomposition} are trivially correct based on the three conditions under which we add a pair~$(D,Q)$ to the construction of the~$r$-min-cut decomposition~$\mc$.
	
	In the following, we show Property~d) of~\Cref{def:cut_decomposition}.
	That is, we have to show that~$C_1 \cap \cut(C_2) = \emptyset$ or~$C_2\subseteq C_1$ for any two distinct vertex pairs~$(C_1, \cut(C_1))$ and~$(C_2, \cut(C_2))$ in~$\mc$.
	To this end, we consider an auxiliary rooted forest using~$\mathcal{D}$ and~$\mathcal{Q}$.\footnote{This forest is not computed by our algorithm and just serves the purpose to prove the desired statement.}
	Each node~$N$ in any tree of the forest corresponds to a vertex set.
	More precisely, the nodes of this forest are exactly the vertex sets from~$\mathcal{D}\cup \mathcal{Q}$.
	For each connected component~$X$ of~$G$, the forest contains a rooted tree with root~$C$.
	Furthermore, the children of a node~$N$ are the connected components of~$G[N]- Q$ and the minimum vertex cut~$Q\in \mathcal{Q}$ of~$G[N]$. 
	Observe that if some node~$P$ is an ancestor of some other node~$N$ then we have~$N\subseteq P$ and if they are unrelated we must have~$N \cap P = \emptyset$.

	Let~$(C_1, \cut(C_1))$ and~$(C_2, \cut(C_2))$ be distinct vertex pairs of~$\mc$.
	We now show that~$C_1\cap \cut(C_2)= \emptyset$ or~$C_2 \subseteq C_1$. 
	Assume towards a contradiction that this is not the case, that is, $C_1\cap \cut(C_2)\neq \emptyset$ and~$C_2 \not\subseteq C_1$. 
	By construction of the $r$-min-cut decomposition~$\mc$, $\cut(C_2)$ is contained in~$\mathcal{Q}$ and thus has no children in its unique tree.
	Thus, $C_1$ is an ancestor of~$\cut(C_2)$. 
	Since~$\cut(C_2)\subseteq N_G(C_2)$, $C_2$ was added to~$\mathcal{D}$ at the earliest when~$\cut(C_2)$ was first discovered as a minimum vertex cut. 
	This implies that~$C_2$ is contained in a subtree which is rooted at some sibling of~$\cut(C_2)$.
	Consequently, $C_2\subseteq C_1$; a contradiction. 
	
	This then implies that an~$r$-min-cut decomposition is an $r$-cut decomposition.
\end{proof}

Note that the computations of the families $\mathcal{D}$ and $\mathcal{Q}$ are independent of the choice of~$r$. 
Consequently, we have to compute these families only once and can reuse them for each value of~$r$.
For each concrete value of~$r$, where we want to compute an $r$-cut decomposition, we can then just select the sets from $\mathcal{D}$ of size at most~$r$. 
To actually use an~$r$-cut decomposition as a speed-up for our implementation, we add the following constraint representing~\Cref{lem:vi_cut_sets_prop_1} from~\Cref{lem:vi_cut_sets} to the ILP for each pair~$(C,\cut(C))\in \mathcal{C}$:
\[ n\cdot \left( |N(C) \setminus \cut(C)| - \sum_{v\in N(C) \setminus \cut(C)} x_v \right) \geq \sum_{v\in C} x_v.\]
If~$|C| = r$ we also add the following constraint representing~\Cref{lem:vi_cut_sets_prop_2}:
\[ n\cdot \left( |N(C) \setminus \cut(C)| - \sum_{v\in N(C) \setminus \cut(C)} x_v \right) \geq |\cut(C)| - \sum_{v\in \cut(C)} x_v.\]
Additionally, if~$N(C) \setminus \cut(C) = \emptyset$, then we do not need to add these constraints and can instead use~\Cref{lem:vi_cut_sets_prop_1,lem:vi_cut_sets_prop_2} during the initial graph reduction.

Finally, to ensure that we can use both our main speed-up strategies simultaneously, we need to verify that there exists a vi-set that not just fulfills the properties of \Cref{lem:vi_cut_sets} but also does not contain redundant vertices. 
In the following, we show that this is the case.

\begin{theorem}\label{thm:vi_all_together}
	Let $G$ be a graph.
	Moreover, for each $r\in [|V(G)|]$, let $\mathcal{C}_r$ be an~$r$-cut decomposition.
	There exists a vi-set~$S$ that fulfills the following conditions:
	\begin{enumerate}
		\item $S$ contains no redundant vertices.
		\item There exists an integer~$r$, such that (i)~$S$ is a minimum-size modulator for $\coc[r](G)$ and (ii)~$S$ together with the~$r$-cut decomposition~$\mathcal{C}_r$ fulfill \Cref{lem:vi_cut_sets_prop_1,lem:vi_cut_sets_prop_2} from \Cref{lem:vi_cut_sets}.
	\end{enumerate}
\end{theorem}
\begin{proof}
Let~$r$ be the largest integer such that there exists a vi-set~$S$, such that~$S$ is a minimum-size modulator for~$\coc[r](G)$.
Due to the definition of vertex integrity and~\Cref{lem:vi_cut_sets}, we can assume without loss of generality that~$\mathcal{C}_r$ fulfills \Cref{lem:vi_cut_sets_prop_1,lem:vi_cut_sets_prop_2} from \Cref{lem:vi_cut_sets} with respect to~$S$.

	If~$S$ contains no redundant vertex, we are done.
	Hence, let us assume that there is a redundant vertex~$v$ in~$S$.
	Consider the vertex set~$S'\coloneqq S \setminus \{v\}$.
	Since~$v$ is a redundant vertex, there is at most one connected component~$C_v$ in~$G-S$ where~$v$ has neighbors in.
	Moreover, since~$S$ is a modulator for~$\coc[r](G)$, $C_v$ has size at most~$r$.
	Hence, $S'$ is a modulator for~$\coc[(r+1)](G)$, since no component of~$G-S'$ has size more than~$r+1$.
	This implies that~$|S'| + \cc(G-S') \leq |S| + \cc(G-S) = \vi(G)$.
	Recall that by definition of~$\vi$, $\vi(G) \leq |S'| + \cc(G-S')$.
	Hence, $|S'| + \cc(G-S') = \vi(G)$, which implies that~$S'$ is (i)~a vi-set and (ii)~a minimum-size modulator for~$\coc[r+1](G)$.
\end{proof}

\section{Running Time Comparison}\todo{Defer to appendix? Add infos requested by Reviewer?}

\begin{table}[t]
\tiny
\caption{Average, median, and 90th percentile running times (in ms) for computing the parameters.}
\label{table-para-statistic-times}
\centering

\begin{tabular}{l r r r}
\toprule
$k$ & avg. time & median time & 90th p. time \\
\midrule
$\Delta$ & 0.0 & 0.0 & 0.0 \\
$h$-index & 0.0 & 0.0 & 0.0 \\
$d$ & 0.4 & 0.0 & 1.0 \\
$c$ & 1.6 & 0.0 & 4.0 \\
$\nd$ & 3.6 & 0.0 & 9.2 \\
$\core$ & 2.0 & 1.0 & 4.0 \\
$\core[3]$ & 1.5 & 1.0 & 4.0 \\
$\mw$ & 2.8 & 2.0 & 5.0 \\
$\gamma$ & 43.3 & 9.0 & 122.2 \\
$\vc$ & 11.8 & 9.0 & 22.4 \\
$\bdd[2]$ & 145.6 & 9.0 & 192.6 \\
$\bdd$ & 193.5 & 12.0 & 211.2 \\
$\sw$ & 122.3 & 17.0 & 284.0 \\
$\nabla$ & 46.2 & 20.0 & 146.2 \\
$\cvd$ & 641.6 & 22.0 & 712.4 \\
$\tdp$ & 154\,108.7 & 43.0 & 16\,446.2 \\
$\fvs$ & 126.3 & 103.0 & 161.8 \\
$\dco$ & 11\,738.8 & 125.0 & 9\,082.4 \\
$\pvc$ & 925.4 & 385.0 & 1\,599.8 \\
$\tw$ & 3\,402.6 & 1\,284.0 & 10\,962.0 \\
$\vi$ & 2\,529\,051.8 & 2\,183.0 & 189\,970.0 \\
\bottomrule
\end{tabular}

\end{table}

Finally, we briefly discuss the time needed to compute the parameters. 
\Cref{table-para-statistic-times} shows the average, median, and 90th percentile running times for the considered parameter on the $59$ instances where all parameters could be computed. Unsurprisingly, the running time for the polynomial-time computable parameters is negligible with the exception of the splitwidth and the Dilworth number which are still easily computable in the given time frame. Among the parameters that are NP-hard to compute, the vertex cover number can be computed quite fast compared to the other parameters of which the currently hardest parameters seem to be~distance to cograph,~treewidth,~treedepth, and vertex integrity.

\begin{figure}
\centering
\includegraphics[scale=0.6]{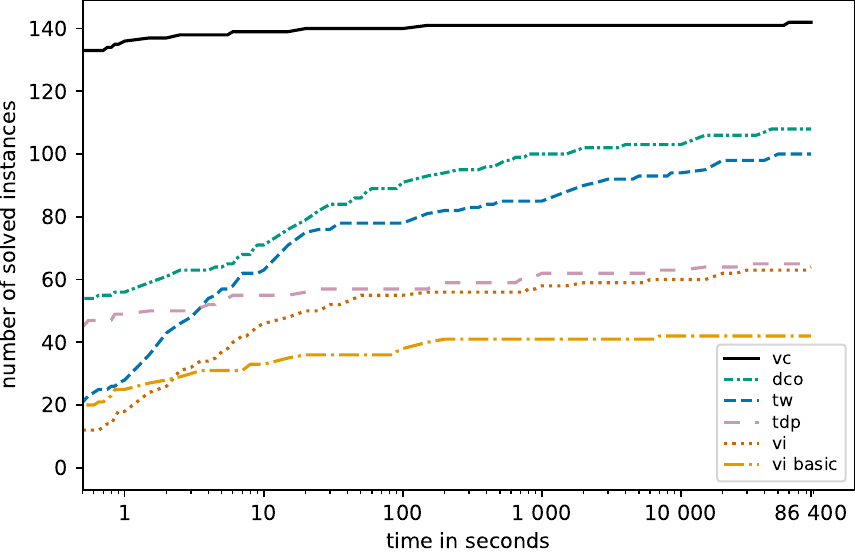}
\caption{Running time comparison between the parameters $\vc$, $\dco$, $\tw$, $\tdp$ and the basic and improved versions of the $\vi$ algorithm.  
For each parameter and each time~$t$ the figure shows on how many of the $144$~instances that parameter could be computed in at most $t$~seconds.}
\label{fig:running_time_plot}
\end{figure}

In \Cref{fig:running_time_plot} we compare these parameters on all $144$ graphs. We also included vertex cover as a contrast since it is the easiest of the parameters that are NP-hard to compute. The plot shows for each time $t$ on how many graphs each parameter could be calculated in at most $t$ seconds.
We can see that it is a lot more difficult to compute~$\dco$,~$\tw$,~$\tdp$, and~$\vi$ compared to~$\vc$ where almost all instances could be solved in less than one second. However, even the~$\dco$, $\tw$ algorithms are quite fast compared to the algorithms for~$\tdp$ and~$\vi$. This shows that it is important to develop faster algorithms for computing these parameters.

In \Cref{fig:running_time_plot} we also compare the two different algorithms for~$\vi$ described in \Cref{sec:vi_algos}. Recall that the algorithm denoted as \texttt{vi}~\texttt{basic} uses a simple ILP-formulation while the algorithm denoted as \texttt{vi} is an improved version with additional reduction rules and ILP constraints.
We see that the improved version of the algorithm is significantly faster than the basic version, which could only solve 42 instances. 
Moreover, each instance solved by \texttt{vi}~\texttt{basic} within the time limit is solved by \texttt{vi} in less than 800~seconds.

\end{document}